\documentclass[12pt]{article}

\usepackage{booktabs}

\usepackage{amssymb,latexsym}
\usepackage{graphicx,amsmath,amsfonts}
\usepackage{comment}
\usepackage{tikz}
\usepackage{pgfplots}

\usepgfplotslibrary{fillbetween}
\usepackage{url}
\usepackage{mathrsfs}
\usepackage{paralist}
\usepackage{hyperref}
\usepackage{nicefrac}
\usepackage{booktabs}
\usepackage{setspace}
\usepackage{etoolbox}
\usepackage{wrapfig}
\usepackage{natbib}
\setcitestyle{authoryear}
\usepackage{amsthm}
\allowdisplaybreaks

\makeatletter

\patchcmd{\NAT@citex}
  {\@citea\NAT@hyper@{%
     \NAT@nmfmt{\NAT@nm}%
     \hyper@natlinkbreak{\NAT@aysep\NAT@spacechar}{\@citeb\@extra@b@citeb}%
     \NAT@date}}
  {\@citea\NAT@nmfmt{\NAT@nm}%
   \NAT@aysep\NAT@spacechar\NAT@hyper@{\NAT@date}}{}{}

\patchcmd{\NAT@citex}
  {\@citea\NAT@hyper@{%
     \NAT@nmfmt{\NAT@nm}%
     \hyper@natlinkbreak{\NAT@spacechar\NAT@@open\if*#1*\else#1\NAT@spacechar\fi}%
       {\@citeb\@extra@b@citeb}%
     \NAT@date}}
  {\@citea\NAT@nmfmt{\NAT@nm}%
   \NAT@spacechar\NAT@@open\if*#1*\else#1\NAT@spacechar\fi\NAT@hyper@{\NAT@date}}
  {}{}

\makeatother

\renewenvironment{description}
  {\list{}{\labelwidth=15pt \leftmargin=15pt
   }}
  {\endlist}

\newcommand{\arxiv}[1]{}

\hypersetup{
colorlinks=true,citecolor=blue!80!black,linkcolor=red!60!black,
pagebackref=true
}

\newtheorem{theorem}{Theorem}

\newtheorem{corollary}{Corollary}
\newtheorem{lemma}{Lemma}

\newtheorem{definition}{Definition}
\newtheorem{example}{Example}
\newtheorem{proposition}{Proposition}

\newcommand{\score}[1]{{{\mathrm{sc}_{#1}}}}
\newcommand{\reals}{\mathbb R}

\newcommand{\naturals}{\mathbb N}

\newcommand{\calA}{\mathcal{A}}

\newcommand{\calR}{\mathcal{R}}
\newcommand{\calF}{\mathcal{F}}
\newcommand{\calG}{\mathcal{G}}
\newcommand{\calM}{\mathcal{M}}

\newcommand{\calP}{\mathcal{P}}
\newcommand{\pos}{{{{\mathrm{pos}}}}}

\newcommand{\toapproval}{{{{\mathit{Appr}}}}}
\newcommand{\toordinal}{{{{\mathit{Rank}}}}}
\newcommand{\most}{{{{\mathit{Bnd}}}}}
\newcommand{\bounded}{\most}
\newcommand{\exact}{{{{\mathit{Reg}}}}}
\newcommand{\regular}{\exact}

\newcommand{\shortcite}[1]{(\citeyear{#1})}

\makeatletter
\newtheorem*{rep@theorem}{\rep@title}
\newcommand{\newreptheorem}[2]{%
\newenvironment{rep#1}[1]{%
 \def\rep@title{#2 \ref{##1}}%
 \begin{rep@theorem}}%
 {\end{rep@theorem}}}
\makeatother
\newreptheorem{theorem}{Theorem}
\newreptheorem{proposition}{Proposition}
\newreptheorem{lemma}{Lemma}
\newreptheorem{corollary}{Corollary}

\newcommand{\boxify}[1]{\vspace{10px}\noindent\fbox{\parbox{0.98\textwidth}{\centering\parbox{0.96\textwidth}{#1}}}\vspace{10px}}
\newcommand{\axiom}[2]{\boxify{\vspace*{3px}\textbf{#1.} #2}}
\newcommand{\axiomset}[1]{\boxify{\vspace*{3px} #1}}

\newcommand{\posf}{{{g}}}

\newcommand{\pow}{\mathscr P}
\newcommand{\set}{\mathit{set}}

\def\argmax{\mbox{argmax}}

\usetikzlibrary{plotmarks}

\begin{document}

\title{Consistent Approval-Based Multi-Winner Rules\footnote{This work was presented at EC'18, the 19th ACM Conference on Economics and Computation; an abstract was published in the proceedings~\citep{ec/LacknerSkowron-consistentmwrules}.}} 

\author{Martin Lackner\\
  TU Wien\\
  Vienna, Austria
  \and 
Piotr Skowron\\
  University of Warsaw\\
  Warsaw, Poland
}
\date{}

\maketitle

\begin{abstract}
This paper is an axiomatic study of consistent approval-based multi-winner rules, i.e., voting rules that select a fixed-size group of candidates based on approval ballots.
We introduce the class of counting rules and provide an axiomatic characterization of this class 
based on the consistency axiom.
Building upon this result, we axiomatically characterize three important consistent multi-winner rules: Proportional Approval Voting, Multi-Winner Approval Voting and the Approval Chamberlin--Courant rule.
Our results demonstrate the variety of multi-winner rules and illustrate three different, orthogonal principles that multi-winner voting rules may represent: individual excellence, diversity, and proportionality.\end{abstract}

\section{Introduction}

A \emph{multi-winner} rule selects a fixed-size set of candidates---a \emph{committee}---based on the preferences of voters.
Multi-winner elections are of importance in a wide range of scenarios, which often fit in, but are not limited to, one of the following three categories \citep{elk-fal-sko-sli:c:multiwinner-rules,FSST-trends}.
The first category contains multi-winner elections aiming for proportional representation. The archetypal example of a multi-winner election is that of selecting a representative body such as a parliament, where a fixed number of seats are to be filled; and these seats are ideally filled so as to proportionally represent the population of voters.
Hence, voting rules used in parliamentary elections typically follow the intuitive principle of \emph{proportionality}, i.e., the chosen subset of candidates should proportionally reflect the voters' preferences.
The second category comprises multi-winner elections with the goal that as many voters as possible should have an acceptable representative in the committee. Consequently, there is no or little weight put on giving voters a second representative in the committee. Ensuring a representative for as many voters as possible can be viewed as an egalitarian objective. This goal may be desirable, e.g., in a deliberative democracy~\citep{ccElection,dryzek2003social}, where it is more important to represent the diversity of opinions in an elected committee rather than to include multiple members representing the same popular opinion.
Another example would be the distribution of facilities such as hospitals in a country, where voters would prefer to have a hospital close to their home but are less interested in having more than one in their vicinity.
Voting rules suitable in such scenarios follow the principle of \emph{diversity}.
The third category contains scenarios where the goal is to choose a fixed number of best candidates and where ballots are viewed as expert judgments. Here, the chosen multi-winner rule should follow the \emph{individual excellence} principle, i.e., to select candidates with the highest total support of the experts.
An example is shortlisting nominees for an award 
where a nomination itself is often viewed as an achievement.
We consider multi-winner rules based on approval ballots, which allow voters to express \emph{dichotomous preferences}. Dichotomous preferences distinguish between approved and disapproved candidates---a dichotomy. 
An approval ballot thus corresponds to a subset of (approved) candidates.
A simple example of an approval-based election can highlight the distinct nature of proportionality, diversity, and individual excellence:
There are 100 voters and 5 candidates $\{a,b,c,d,e\}$: 66 voters approve the set $\{a,b,c\}$, 33 voters approve $\{d\}$, and one voter approves $\{e\}$. Assume we want to select a committee of size three. If we follow the principle of proportionality, we could choose, e.g., $\{a,b,d\}$; this committee closely reflects the proportions of voter support. If we aim for diversity and do not consider it important to give voters more than one representative, we may choose the committee $\{a,d, e\}$: it contains one approved candidate of every voter. The principle of individual excellence aims to select the strongest candidates: $a$, $b$, and $c$ have most supporters and are thus a natural choice, although the opinions of 34 voters are essentially ignored. We see that these three intuitive principles give rise to very different committees. 
In this paper, we will explore these principles in a formal, mathematical framework.

For single-winner rules, one distinguishes between social welfare functions, i.e., voting rules that output a ranking of candidates, and social choice functions, i.e., rules that output a single winner or a set of tied winners. For multi-winner rules, an analogous classification applies:
we distinguish between \emph{approval-based committee (ABC) ranking rules}, which produce a ranking of all committees, and \emph{ABC choice rules}, which output a set of winning committees.
In this work, we mainly focus on ABC ranking rules since this model is more versatile. In particular, with ABC ranking rules one can easily combine the societal evaluation of committees with additional requirements one would like to impose on the structure of the committee. E.g., suppose the goal is to select a committee subject to certain diversity constraints (such as an equal number of men and women). In such scenarios a ranking rule can be applied directly: among the committees that satisfy the diversity constraint, one can simply select the committee that is ranked highest. Although our focus is mainly one ABC ranking rules, we present a technique that allows us to translate some of our axiomatic characterizations to the framework of ABC choice rules.

While axiomatic questions are well explored for both social choice and social welfare functions, far fewer results are known for multi-winner rules (we provide an overview of the related literature in Section~\ref{sec:related_work}).
Such an axiomatic exploration of multi-winner rules is essential if one wants to choose a rule in a principled way.
Axiomatic characterizations of multi-winner rules are of crucial importance because those rules may have very different objectives: as we have seen in the example, proportionality, diversity, and individual excellence may be conflicting goals.

The main goal of this paper is to explore the class of \emph{consistent} ABC ranking rules. An ABC ranking rule is consistent if the following holds: if two disjoint societies decide on the same set of candidates and if both societies prefer committee $W_1$ to a committee $W_2$, then the union of these
two societies should also prefer $W_1$ to $W_2$.
This is a straightforward adaption of consistency as defined for single-winner rules by Smith~\shortcite{smi:j:scoring-rules} and Young~\shortcite{young74}.
Our results highlight the diverse landscape of consistent multi-winner rules and their defining and widely varying properties.

\subsection{Main results}
The first main result of this paper is an axiomatic characterization of \emph{ABC counting rules}, which are a subclass of ABC ranking rules.
ABC counting rules are informally defined as follows:
given a real-valued function $f(x,y)$ (the so-called \emph{counting function}), a committee $W$ receives a score of $f(x,y)$ from every voter for whom committee $W$ contains $x$ approved candidates and who approves~$y$ candidates in total; the ABC counting rule implemented by $f$ ranks committees according to the sum of scores obtained from all voters. 
We obtain the following characterization.
\newcommand{\thmcharacterizationWelfareFunctions}{An ABC ranking rule is an ABC counting rule if and only if it satisfies symmetry, consistency, weak efficiency, and continuity.}
\begin{theorem}\label{thm:characterizationWelfareFunctions}
\thmcharacterizationWelfareFunctions
\end{theorem}%
The axioms used in this theorem can be intuitively described as follows:
We say that a rule is symmetric if the names of voters and candidates do not affect the result of an election.
Weak efficiency informally states that candidates that no one approves cannot be ``better'' committee members than candidates that are approved by some voter.
Continuity is a more technical condition that states that a sufficiently large majority can dictate a decision.
As weak efficiency is satisfied by every sensible multi-winner rule and continuity typically only rules out the use of certain tie-breaking mechanisms \citep{smi:j:scoring-rules,young74,you:j:scoring-functions}, Theorem~\ref{thm:characterizationWelfareFunctions} essentially implies that ABC counting rules correspond to symmetric and consistent ABC ranking rules.
Furthermore, we show that the set of axioms used to characterize ABC counting rules is minimal.

Theorem~\ref{thm:characterizationWelfareFunctions} gives a powerful technical tool that allows us to obtain further axiomatic characterizations of more specific rules. Indeed, building upon this result, we  explore the space of ABC counting rules, and obtain our second main result---the axiomatic explanation of the differences between three important rules: Multi-Winner Approval Voting (AV), Proportional Approval Voting (PAV), and Approval Chamberlin--Courant (CC), which are defined by the following counting functions:
\begin{align*}
f_\text{AV}(x,y)=x; \quad\quad\quad f_\text{PAV}(x,y)=\sum_{i=1}^x\nicefrac{1}{i}; \quad\quad\quad f_\text{CC}(x,y) = \begin{cases}0 & \text{if }x = 0,\\1 & \text{if }x\geq 1.\end{cases}
\end{align*}
Note that these three specific example of counting functions do not depend on the parameter $y$ as they belong to the class of Thiele methods; we discuss this fact in Section~\ref{subsec:counting-rule-def}.
These three well-known rules are prime examples of multi-winner systems following the principle of individual excellence, proportionality, and diversity, respectively. Our results imply that the differences between these three rules can be understood by studying how these rules behave when viewed as \emph{apportionment methods}.
Apportionment methods are a well-studied special case of approval-based multi-winner voting, where the set of candidates can be represented as a disjoint union of groups (intuitively, each group can be viewed as a political party), and where each voter approves all candidates within one of these groups (which can be viewed as voting for a single party)---we refer to preference profiles that can be represented in such way as \emph{party-list profiles}. For these mathematically much simpler profiles it is easier to formalize the principles of individual excellence, proportionality and diversity:
\begin{description}
\item[Disjoint equality] states that if each candidate is approved by at most one voter, then any committee consisting of approved candidates is a winning committee. One can argue that the principle of individual excellence implies disjoint equality: if every candidate is approved only once, then every approved candidate has the same support, their ``quality'' cannot be distinguished, and hence all approved candidates are equally well suited for selection.
\item[D'Hondt proportionality] defines a way in which parliamentary seats are assigned to parties in a proportional manner. The D'Hondt method (also known as Jefferson method) is one of the most commonly used methods of apportionment in parliamentary elections.
\item[Disjoint diversity] states that as many parties as possible should receive one seat and, if necessary, priority is given to stronger parties. 
Disjoint diversity is implied by apportionment methods such as Huntington-Hill, Dean, or Adams.
\end{description}
We show that Multi-Winner Approval Voting is the only ABC counting rule that satisfies disjoint equality, Proportional Approval Voting is the only ABC counting rule satisfying D'Hondt proportionality and that Approval Chamberlin--Courant is the only ABC counting rule satisfying disjoint diversity.
Together with Theorem~\ref{thm:characterizationWelfareFunctions}, these results lead to unconditional axiomatic characterizations of AV, PAV, and CC. In particular, our results show that Proportional Approval Voting is essentially the only consistent extension of the D'Hondt method to the more general setting where voters decide on individual candidates rather than on parties.

Our results illustrate the variety of ABC ranking rules, even within the class of consistent ABC ranking rules. This variety is due to their defining counting function $f(x,y)$; see Figure~\ref{fig:countingfcts} for a visualization.
Our results indicate that counting functions that have a larger slope than $f_\text{PAV}$ put more emphasis on majorities and thus become less egalitarian, whereas counting functions that have a smaller slope than $f_\text{PAV}$ treat minorities preferentially and thus approach the Approval Chamberlin--Courant rule.

\begin{figure}
\begin{center}
\begin{tikzpicture}[y=.6cm, x=1.2cm,font=\sffamily]

\pgfdeclarelayer{bg}
\pgfsetlayers{bg,main}
    \draw (0,0) -- coordinate (x axis mid) (7.5,0);
        \draw (0,0) -- coordinate (y axis mid) (0,8.5);
        \foreach \x in {0,...,7}
             \draw (\x,1pt) -- (\x,-3pt)
            node[anchor=north] {\x};
        \foreach \y in {0,2,...,8}
             \draw (1pt,\y) -- (-3pt,\y) 
                 node[anchor=east] {\y}; 
    \node[below=0.8cm] at (x axis mid) {number of approved candidates in committee ($x$)};
    \node[rotate=90, above=0.8cm] at (y axis mid) {score $f(x,.)$};

    \draw plot[mark=*, mark options={fill=white}] 
        file {av.data};
    \draw plot[mark=triangle*, mark options={fill=white} ] 
        file {pav.data};
    \draw plot[mark=square*]
        file {cc.data}; 
    \path[name path = ub] plot[] 
        file {ub.data}; 
    \path[name path = lb] plot[] 
        file {lb.data}; 
    \begin{pgfonlayer}{bg}
        \fill [black!20,
          intersection segments={
            of=lb and ub,
            sequence={L*-- R*[reverse]}
          }];
    \end{pgfonlayer}
    
    \node[right=0.2cm] at (7,1) {$f_\text{CC}$};
    \node[right=0.2cm] at (7,2.717857143) {$f_\text{PAV}$};
    \node[right=0.2cm] at (7,7) {$f_\text{AV}$};        

    \begin{scope}[shift={(0.5,6.75)}] 
    \draw[yshift=2\baselineskip] (0,0) -- 
        plot[mark=*, mark options={fill=white}] (0.25,0) -- (0.5,0) 
        node[right]{Multi-Winner Approval Voting};
    \draw[yshift=\baselineskip] (0,0) -- 
        plot[mark=triangle*, mark options={fill=white}] (0.25,0) -- (0.5,0)
        node[right]{Proportional Approval Voting};
    \draw[yshift=0\baselineskip] (0,0) -- 
        plot[mark=square*, mark options={fill=black}] (0.25,0) -- (0.5,0)
        node[right]{Approval Chamberlin--Courant};
    \end{scope}
\end{tikzpicture}
\caption{Different counting functions and their corresponding ABC counting rules. Counting functions outside the gray area fail the lower quota axiom; see Section~\ref{sec:lower-quota} for a formal statement.}
\label{fig:countingfcts}
\end{center}
\end{figure}
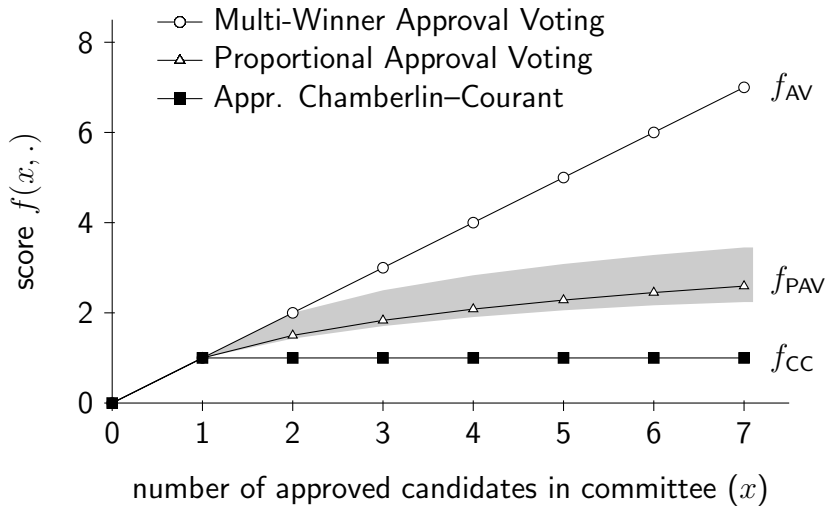

We furthermore present several extensions of our main results. 
First, we show that counting functions that are not ``close'' to $f_\text{PAV}$ (all those not contained in the gray area around $f_\text{PAV}$ in Figure~\ref{fig:countingfcts}) implement ABC ranking rules that violate a rather weak form of proportionality called lower quota.
Second, we show that our characterization of PAV can be generalized to a broader class of ABC counting rules: given a shape of proportionality on party-list profiles, represented as a specific divisor apportionment method%
, we show that there is a unique symmetric, consistent and continuous ABC ranking rule that guarantees this kind of proportionality. 
Third, we prove that the characterizations of Proportional Approval Voting, Approval Chamberlin--Courant, and the aforementioned generalization to arbitrary divisor methods also hold for ABC choice rules.

Finally, we note that our main theorem, Theorem~\ref{thm:characterizationWelfareFunctions}, can also be used to obtain further characterizations that are not based on proportionality or apportionment. For example, independence of irrelevant alternatives characterizes the classes of Thiele methods, and a variant of strategyproofness yields an alternative characterizations of Multi-Winner Approval Voting~\citep{lac-sko:t:multiwinner-strategyproofness}.

\subsection{Related Work}\label{sec:related_work}

Electing a representative body such as a parliament is perhaps the most classic example of a multi-winner election; we refer to the books of Farrell~\shortcite{Far11}, and Renwick and Pilet~\shortcite{renwick16} for an overview of multi-winner elections in a political context. In recent years, there has been an emerging interest in multi-winner elections from the computer science community.
In this context, multi-winner election rules have been analyzed and applied in a variety of scenarios: 
personalized recommendation and advertisement~\citep{budgetSocialChoice,bou-lu:c:value-directed}, 
group recommendation~\citep{owaWinner}, 
diversifying search results~\citep{proprank}, improving genetic algorithms~\citep{fal-saw-sch-smo:c:multiwinner-genetic-algorithms}, and the broad class of facility location problems~\citep{far-hek:b:facility-location,owaWinner}. In all these settings, multi-winner voting either appears as a core problem itself or can help to improve or analyze mechanisms and algorithms.
For an overview of this literature we refer the reader to a recent survey by Faliszewski, Skowron, Slinko, and Talmon~\shortcite{FSST-trends}.

The most important axiomatic concept in our study is \emph{consistency}.
In the context of social welfare functions, this axiom informally states that if two disjoint societies both prefer candidate $a$ over candidate $b$, then the union of these two societies should also prefer $a$ over $b$, i.e., consistency refers to consistent behavior with respect to varying populations, referred to as population-consistency in other contexts.
Smith~\shortcite{smi:j:scoring-rules} and Young~\shortcite{young74} independently introduced this axiom and characterized the class of positional scoring rules as the only social welfare functions that satisfy symmetry, consistency, and continuity. %
Subsequently, Young~\shortcite{you:j:scoring-functions} proved an analogous result for social choice functions, i.e., rules that return the set of winning candidates.
Further, Myerson~\shortcite{Myerson1995} and Pivato~\shortcite{pivato2013variable} characterized positional scoring rules with the same set of axioms but without imposing any restriction on the input of rules, i.e., ballots are not restricted to be a particular type of order.
Extensive studies led to further, more specific, characterizations of consistent voting rules~\citep{che-sha:j:scoring-rules,merlinAxiomatic}.
Also in probabilistic social choice, consistency is an important concept: Brandl, Brandt, and Seedig~\shortcite{Bran13a} characterize Fishburn's rule of maximal lotteries~\shortcite{Fish84a} via two consistency axioms.

The impressive body of axiomatic studies shows that single-winner voting is largely well-understood and characterized.
Axiomatic properties of \emph{multi-winner rules} are considerably fewer in number. Debord~\shortcite{deb:j:k-borda} characterized the $k$-Borda rule using similar axioms as Young~\shortcite{youngBorda}. Felsenthal and Maoz~\shortcite{fel-mao:j:norms} and Elkind, Faliszewski, Skowron, and Slinko~\shortcite{elk-fal-sko-sli:c:multiwinner-rules} formulated a number of axiomatic properties of multi-winner rules, and analyzed which rules satisfy these axioms; however, they do not obtain axiomatic characterizations. Elkind~et~al.~\shortcite{elk-fal-sko-sli:c:multiwinner-rules} also defined the class of committee scoring rules, which aims at generalizing single-winner positional scoring rules to the multi-winner setting. This broad class contains, among others, the Chamberlin--Courant rule~\citep{ccElection}. In recent work, Skowron, Faliszewski, and Slinko~\shortcite{skowron2019axiomatic} showed that the class of committee scoring rules admits a similar characterization as their single-winner counterparts. Faliszewski, Skowron, Slinko, and Talmon~\shortcite{fal-sko-sli-tal:c:top-k-counting,fal-sko-sli-tal-tal:j:hierarchy-committee} further studied the internal structure of committee scoring rules and characterized several multi-winner rules within this class.

The aforementioned characterization result of committee scoring rules \citep{skowron2019axiomatic} plays a major role in the proof of Theorem~\ref{thm:characterizationWelfareFunctions}.
Let us briefly discuss this connection. 
Committee scoring rules are multi-winner voting rules that accept preferences in the form of linear orders as input and output a ranking of committees (a definition can be found in Appendix~\ref{sec:app:countingrules}).
The main difference to ABC ranking rules is thus the type of preferences.
It is important to note that the characterization of \citeauthor{skowron2019axiomatic} only holds for linear orders and not for weak orders, hence ABC ranking rules are not covered by the result.
On the contrary, it requires substantial work to prove the exact relation between these two classes so that results can be transferred from one class to the other.
Furthermore, in our work we characterize specific voting rules based, e.g., on proportionality axioms.
These results are a strength of our model, as proportionality can be formulated much more clearly for approval preferences.
This is because 
proportionality is well-understood for party-list elections; the corresponding mathematical problem is called the \emph{apportionment problem}. It typically arises when allocating seats to political parties based on the number of votes.
We explain in Section~\ref{sec:proportionality} how ABC ranking rules can be seen as generalizations of apportionment methods.
For an overview of the literature on apportionment we refer the reader to the comprehensive books by Balinski and Young~\shortcite{BaYo82a} and by Pukelsheim~\shortcite{Puke14a}.

The concept of proportionality in arbitrary multi-winner elections (i.e., in the absence of parties) is more elusive.
The first study of proportional representation in multi-winner voting dates back to Black~\shortcite{bla:b:polsci:committees-elections}, who informally defined proportionality as the ability to reflect shades of a society's political opinion in the elected committee. 
Later, Dummett~\shortcite{dum:b:voting} proposed an axiom of proportionality for multi-winner rules that accept linear orders as input; it is based on the top-ranked candidates in voters' rankings.
An insightful discussion on the need for notions of proportionality that are applicable to linear order preferences can be found in the seminal work of Monroe~\shortcite{monroeElection}; he referred to such concepts as \emph{fully proportional representation} since they are to take ``full'' preferences into account. 
More recently, axiomatic properties for approval-based rules have been proposed that aim at capturing the concept of proportional representation~\citep{justifiedRepresenattion, pjr17}.
This body of research demonstrates that the concept of proportionality can be sensibly defined and discussed in the context of multi-winner rules, even though this setting is more intricate and mathematically complex than the party-list setting.
It is noteworthy that the results in our paper (in particular Theorem~\ref{thm:thiele-abc-rule-characterization}) show that for obtaining axiomatic characterizations it is in general not necessary to rely on proportionality definitions considering the full domain; proportionality defined in the restricted party-list setting---i.e., proportionality as defined for the apportionment problem---may be sufficient for characterizing multi-winner rules.

There also exist more critical works raising arguments against the use of classical (i.e., linear) proportionality for electing representative assemblies. One class of critical arguments arises from the analysis of issues related to the concept of voting power~\citep{FesMac98}. The second main objection comes from the analysis of probabilistic models describing how the decisions made by the elected committee map to the satisfaction of individual voters participating in the process of electing the committee~\citep{RePEc:ucp:jpolec:doi:10.1086/670380}. This kind of analysis often gives arguments in favor of other concepts such as ``degressive'' proportionality. In Section~\ref{subsec:divisor} we will explain that our results can be easily extended to apply to such other forms of proportionality.

\subsection{Structure of the Paper}

This paper is structured as follows: We briefly state preliminary definitions in Section~\ref{sec:prel}.
Section~\ref{sec:abc-counting} contains a formal introduction of ABC counting rules,
their defining set of axioms, and the statement of our 
main technical tool (Theorem~\ref{thm:characterizationWelfareFunctions}).
In Section~\ref{sec:proportionality}, we discuss and prove the axiomatic characterization of Proportional Approval Voting based on D'Hondt proportionality, %
the characterization of Multi-Winner Approval Voting %
based on disjoint equality,
and the characterization of Approval Chamberlin--Courant via disjoint diversity. %
In Section~\ref{sec:extensions} we extend these characterizations: 
we make the statement precise that only functions ``close'' to $f_\text{PAV}$ can satisfy the lower quota axiom, extend the characterization of PAV to arbitrary divisor methods, and translate our results to ABC choice rules.
Finally, in Section~\ref{sec:concl} we summarize the big picture of this paper.
As this paper contains a large number of proofs, we have moved substantial parts into appendices. 
Appendix~\ref{sec:app:countingrules} contains the main technical and most complex part of this paper, the proof of Theorem~\ref{thm:characterizationWelfareFunctions}.
Appendix~\ref{sec:app:proof-details} contains further omitted proofs, including proofs of the main theorems of Section~\ref{sec:proportionality}.
\arxiv{Appendix~\ref{sec:appendix-disjointequ} contains a technical note on the disjoint equality axiom as used in Theorem~\ref{thm:approval-characterizationB}.}
In Appendix~\ref{sec:app:choice_rules}, we show how to translate some of our results from the setting of ABC ranking rules to ABC choice rules.

\section{Preliminaries}
\label{sec:prel}

We write $[n]$ to denote the set $\{1,\dots,n\}$ and $[i,j]$ to denote $\{i,i+1,\dots,j\}$ for $i\leq j \in \naturals$.
For a set $X$, let $\pow(X)$ denote the powerset of $X$, i.e., the set of subsets of $X$.
Further, for each $\ell$ let $\pow_{\ell}(X)$ denote the set of all size-$\ell$ subsets of $X$.
A weak order of $X$ is a binary relation that is transitive and complete (all elements of $X$ are comparable), and thus also reflexive.
A linear order is a weak order that is antisymmetric.
We write $\mathscr W(X)$ to denote the set of all weak orders of $X$ and $\mathscr L(X)$ to denote the set of all linear orders of $X$.

\medskip \noindent
\textbf{Approval profiles.}
Let $C = \{c_1, \ldots, c_m\}$ be a set of candidates.
We identify voters with natural numbers, i.e., $\naturals$ is the set of all possible voters. For each finite subset of voters $V=\{v_1,\dots,v_n\} \subset \naturals$, an \emph{approval profile of $V$} is a function from $V$ to $\pow(C)$; we write $A = (A(v_1), \ldots, A(v_n))$ as a short-form for this function.
For $v \in V$, let $A(v) \subseteq C$ denote the subset of candidates approved by voter $v$.
We write $\calA(C,V)$ to denote the set of all possible approval profiles over $V$ and $\calA(C) = \bigcup_{\text{finite } V \subset \naturals} \calA(C,V)$ to be the set of all approval profiles (for the fixed candidate set~$C$).
Given a permutation $\sigma\colon C\to C$ and an approval profile $A\in\calA(C,V)$, we write $\sigma(A)$ to denote the profile $(\sigma(A(v_1)), \dots, \sigma(A(v_n)))$.

Let $V=\{v_1,\dots,v_s\}\subseteq \naturals$ and $V'=\{v'_1,\dots,v'_{t}\}\subseteq \naturals$.
Further, let $A\in\calA(C,V)$ and $A'\in\calA(C,V')$.
If $V$ and $V'$ do not intersect, we write $A + A'$ to denote the profile 
$B\in \calA(C,V\cup V')$ defined as $B=(A(v_1),\dots,A(v_s),A'(v'_1),\dots,A'(v'_{t}))$.
If $V$ and $V'$ intersect, we relabel the voters to $V''=[1,s+t]$ and define $B\in \calA(C,V'')$ analogously.
For a positive integer $n$, we write $nA$ to denote $A+A+\dots+A$, $n$ times.

\medskip \noindent
\textbf{Approval-based committee ranking rules.}
Let $k$ denote the desired size of the committee to be formed.
We refer to elements of $\pow_k(C)$ as \emph{committees}.
Throughout the paper, we assume that both $k$ and $C$ (and thus $m$) are arbitrary but fixed.
Furthermore, to avoid trivialities, we assume $k<m$.

An \emph{approval-based committee ranking rule (ABC ranking rule)}, $\calF \colon \calA(C) \to \mathscr W(\pow_{k}(C))$, maps approval profiles to weak orders over committees.
Note that $C$ and $k$ are parameters for ABC ranking rules but since we assume that $C$ and $k$ are fixed, we omit them to alleviate notation.
For an ABC ranking rule $\calF$ and an approval profile $A$, we write $\succeq_{\calF(A)}$ to denote the weak order $\calF(A)$. For $W_1,W_2\in\pow_k(C)$, we write $W_1 \succ_{\calF(A)} W_2$ if $W_1 \succeq_{\calF(A)} W_2$ and not $W_2 \succeq_{\calF(A)} W_1$, and we write $W_1 \sim_{\calF(A)} W_2$ if $W_1 \succeq_{\calF(A)} W_2$ and $W_2 \succeq_{\calF(A)} W_1$.
A committee is a \emph{winning committee} if it is a maximal element with respect to $\succeq_{\calF(A)}$.

An \emph{approval-based committee choice rule (ABC choice rule)}, $\calF \colon \calA(C) \to \pow(\pow_k(C))\setminus\{\emptyset\}$, maps approval profiles to sets of committees, again referred to as \emph{winning committees}.
As before, $C$ and $k$ are parameters for ABC choice rules but we omit them from our notation.
Note that each ABC ranking rule naturally defines an ABC choice rule by returning all top-ranked committees.
In contrast, ABC choice rules do not immediately translate to ABC ranking rules, since the relative ranking of losing committees is not known.

An ABC ranking rule is \emph{trivial} if for all $A\in\calA(C)$ and $W_1,W_2\in\pow_k(C)$ it holds that $W_1 \sim_{\calF(A)} W_2$. 
An ABC choice rule is \emph{trivial} if for all $A\in\calA(C)$ it holds that $\calF(A)=\pow_k(C)$.

Let us now list some important examples of ABC ranking rules. For some of these rules it was already mentioned in the introduction that they belong to the class of ABC counting rules; we discuss this classification in detail in Section~\ref{sec:abc-counting} and also give their defining counting functions. The definitions provided here are more standard and do not use counting functions.

\begin{description}
\item[Multi-Winner Approval Voting (AV).] In AV each candidate $c \in C$ obtains one point from each voter who approves of $c$.
The AV-score of a committee $W$ is the total number of points awarded to members of $W$, i.e., $\sum_{v \in V}|A(v) \cap W|$.
Multi-Winner Approval Voting considered as an ABC ranking rule ranks committees according to their score; AV considered as an ABC choice rule outputs all committees with maximum AV-scores.

\item[Thiele Methods.] In 1895, Thiele~\shortcite{Thie95a} proposed a number of ABC ranking rules that can be viewed as generalizations of Multi-Winner Approval Voting. Consider a sequence of weights $w = (w_1, w_2, \ldots, w_k)$ and define the $w$-score of a committee $W$ as $\sum_{v \in V} \sum_{j = 1}^{|W \cap A(v)|}w_j$, i.e., if voter $v$ has $x$ approved candidates in $W$, $W$ receives a score of $w_1+w_2+\dots+w_x$ from $v$. The committees with highest $w$-score are the winners according to the {$w$-Thiele} method. Thiele methods can also be viewed as ABC ranking rules, where committees are ranked according to their score.
\end{description}

Thiele methods form a remarkably general class of multi-winner rules: apart from Multi-Winner AV which is defined by the weights $w_{\mathrm{AV}} = (1, 1, 1, \ldots)$, PAV and CC also fall into this class.
Thiele methods are also a general class in the sense that they contain both polynomial-time computable rules (such as AV) and NP-hard rules (such as PAV or CC)~\citep{owaWinner,azi-gas-gud-mac-mat-wal:c:multiwinner-approval}.

\begin{description}
\item[Proportional Approval Voting (PAV).] PAV was first proposed by Thiele~\shortcite{Thie95a}; it was later reinvented by Simmons~\shortcite{kil-handbook}, who introduced the name ``proportional approval voting''. PAV is a Thiele method defined by the weights $w = (1,\nicefrac{1}{2},\nicefrac{1}{3},\dots)$. These weights being harmonic numbers guarantee proportionality---in contrast to Multi-Winner Approval Voting, which is not a proportional method. The Proportionality of PAV is illustrated in the example in Introduction.

\item[Approval Chamberlin--Courant (CC).] Also the Approval Chamberlin--Courant rule was suggested and recommended by Thiele~\shortcite{Thie95a}. It closely resembles the Chamberlin--Courant~\shortcite{ccElection} rule, which was originally defined for ordinal preferences but can easily be adapted to the approval setting.
The Approval Chamberlin--Courant rule is a Thiele method defined by the weights $w_{\mathrm{CC}} = (1,0,0,\dots)$.
Consequently, CC chooses committees so as to maximize the number of voters who have at least one approved candidate in the winning committee.

\end{description}

\section{ABC Counting Rules}
\label{sec:abc-counting}

In this section we define a new class of multi-winner rules, called ABC counting rules. ABC counting rules can be viewed as an adaptation of positional scoring rules~\citep{smi:j:scoring-rules,young74} to the world of approval-based multi-winner rules.
Furthermore, ABC counting rules can be viewed as analogous to the class of (multi-winner) committee scoring rules as introduced by Elkind~et~al.~\shortcite{elk-fal-sko-sli:c:multiwinner-rules} but defined for approval ballots instead of ranked ballots.

After formally defining ABC counting rules and introducing some basic axioms, we will present our main technical result: an axiomatic characterization of the class of ABC counting rules. This result forms the basis for our subsequent axiomatic analysis.

\subsection{Defining ABC Counting Rules}
\label{subsec:counting-rule-def}

A \emph{counting function} is a mapping $f \colon [0,k]\times[0,m] \to \reals$ satisfying $f(x,y)\geq f(x',y)$ for $x\geq x'$.
The intuitive meaning is that $f(x, y)$ denotes the score that a committee $W$ obtains from a voter that approves of $x$ members of $W$ and $y$ candidates in total.
Let $A\in\calA(C,V)$.
We define the score of $W$ in $A$ as 
\begin{align}
\score{f}(W, A) = \sum_{v \in V} f(|A(v) \cap W|, |A(v)|).\label{eq:score}
\end{align}

We say that a counting function $f$ \emph{implements} an ABC ranking rule $\calF$ if for every $A\in \calA(C)$ and $W_1,W_2\in \pow_k(C)$, 
\begin{align*}
\score{f}(W_1,A) > \score{f}(W_2,A)\quad \text{ if and only if }\quad W_1 \succ_{\calF(A)} W_2.
\end{align*}
Analogously, we say that a counting function $f$ \emph{implements} an ABC choice rule $\calF$ if  for every $A\in \calA(C)$, we have that
$\calF(A)=\argmax_{W\in \pow_k(C)}\score{f}(W,A)$,
i.e., $\calF$ returns all committees with maximum score.
An ABC (winner) rule $\calF$ is an \emph{ABC counting rule} if there exists a counting function~$f$ that implements $\calF$.

Several ABC ranking rules that we introduced earlier are ABC counting rules.
As we have seen in the introduction, AV, PAV, and CC can be implemented by the following counting function:
\begin{align*}
f_\text{AV}(x,y)=x\text{,}  \qquad\quad f_\text{PAV}(x,y)=\sum_{i=1}^x\nicefrac{1}{i}\text{,}  \qquad\quad  f_\text{CC}(x,y) = \begin{cases}0 & \text{if }x = 0,\\1 & \text{if }x\geq 1.\end{cases}
\end{align*}
Further, ABC counting rules include rules such as 
Constant Threshold Methods~\citep{fishburn2004approval} and Satisfaction Approval Voting~\citep{BrKi14a}, implemented by
\begin{align*}
f_\text{CT}(x,y) = \begin{cases}0 & \text{if }x < t,\\1 & \text{if }x\geq t\end{cases}
\qquad\text{ and }\qquad
f_\text{SAV}(x,y) = \frac{x}{y}\text{.}
\end{align*}

Note that only Satisfaction Approval Voting is implemented by a counting function depending on $y$. As can easily be verified, Thiele methods are exactly those ABC counting rules that can be implemented by a counting function not dependent on $y$: the counting function $f$ defined by $w$-Thiele is $f(x,y)=w_1+\dots+w_x$ and, conversely, every counting function $f:\{0,1,\dots,k\}\to \reals$ can be written as $f(x)=\sum_{j=1}^x w_j$ for some sequence of real weights $(w_1,\dots,w_k)$.

It is apparent that not the whole domain of a counting rule is relevant; consider for example $f(2,1)$ or $f(0,m)$---these function values will not be used in the score computation of any committee, cf.\ Equation~\eqref{eq:score}.
The following proposition provides a tool for showing that two counting rules are equivalent.
It shows which part of the domain of counting rules is relevant and that affine transformations yield equivalent rules.

\newcommand{\propcountingfunctionsequivalent}{Let $D_{m,k}= \{(x,y) \in [0,k]\times[0,m-1]: x\leq y \wedge k-x\leq m-y\}$ and let $f,g$ be counting functions.
If there exist $c\in \reals$ and $d\colon [m]\to \reals$ such that $f(x,y)=c\cdot g(x,y)+d(y)$ for all $x,y\in D_{m,k}$ then $f,g$ implement the same ABC counting rule, i.e., for all approval profiles $A\in \calA(C,V)$ and committees $W_1,W_2\in\pow_k(C)$ it holds that $W_1\succ_{f(A)} W_2$ if and only if $W_1\succ_{g(A)} W_2$.}
\begin{proposition}
\propcountingfunctionsequivalent
\label{prop:counting-functions-equivalent}
\end{proposition}

\subsection{Basic Axioms}
\label{subsec:basic-axioms}

In this section, we discuss formal definitions of the axioms used for our main characterization result (Theorem~\ref{thm:characterizationWelfareFunctions}). 
All axioms are natural and straightforward adaptations of the respective properties of single-winner election rules and 
will be stated for ABC ranking rules. In Appendix~\ref{sec:app:choice_rules}, where we extend some of our results to ABC choice rules, we explain how these axioms should be formulated for ABC choice rules.
Similar axioms have also been considered in the context of multi-winner rules based on linear-order preferences~\citep{elk-fal-sko-sli:c:multiwinner-rules, skowron2019axiomatic}.

Anonymity and neutrality enforce perhaps the most basic fairness requirements for voting rules. Anonymity is a property which requires that all voters are treated equally, i.e., the result of an election does not depend on particular names of voters but only on votes that have been cast. In other words, under anonymous ABC ranking rules, each voter has the same voting power. Neutrality is similar, but enforces equal treatment of candidates rather than of voters. 

\axiomset{
  \textbf{Anonymity.} An ABC ranking rule $\calF$ is \emph{anonymous} if for
  $V, V' \subset \naturals$ such that $|V| = |V'|$, for each bijection $\rho: V \to V'$, and for $A \in \calA(C,V)$ and $A' \in \calA(C,V')$
  such that $A(v) = A'(\rho(v))$ for each $v \in V$,
  it holds that $\calF(A) = \calF(A')$.

\medskip\noindent
  \textbf{Neutrality.}  An ABC ranking rule $\calF$ is \emph{neutral} if for each bijection
  $\sigma\colon C\to C$ and $A, A' \in \calA(C,V)$ with $\sigma(A) = A'$ it holds for $W_1,W_2\in \pow_k(C)$ that $W_1\succeq_{\calF(A)} W_2$ if and only if $\sigma(W_1)\succeq_{\calF(A')} \sigma(W_2)$.
}

Due to their analogous structure and similar interpretations, anonymity and neutrality are very often considered together, and jointly referred to as symmetry.

\axiom{Symmetry}{An ABC ranking rule is \emph{symmetric} if it is anonymous and neutral.}

Consistency was first introduced in the context of single-winner rules by Smith~\shortcite{smi:j:scoring-rules} and then adapted by Young~\shortcite{young74}. In the world of single-winner rules, consistency is often considered to be \emph{the} axiom that characterizes positional scoring rules. Similarly, consistency played a crucial role in the recent characterization of committee scoring rules~\citep{skowron2019axiomatic}, which can be considered the equivalent of positional scoring rules in the multi-winner setting.
Consistency is also the main ingredient for our axiomatic characterization of ABC counting rules. 

\axiom{Consistency}{
  An ABC ranking rule $\calF$ is \emph{consistent} if for disjoint $V,V'\subset \naturals$, $A\in \calA(C, V)$, $A'\in \calA(C, V')$, and 
  $W_1, W_2 \in \pow_k(C)$, it holds that
  \begin{enumerate}[(i)]  
  \item if $W_1 \succ_{\calF(A)} W_2$ and $W_1
  \succeq_{\calF(A')} W_2$, then $W_1 \succ_{\calF(A+A')} W_2$, and
  \item if $W_1 \succeq_{\calF(A)} W_2$ and $W_1
  \succeq_{\calF(A')} W_2$, then $W_1 \succeq_{\calF(A + A')} W_2$.
  \end{enumerate}  
}

Next, we describe a weak efficiency axiom, which captures the intuition that candidates approved by no one are undesirable.

\axiom{Weak efficiency}{
  An ABC ranking rule $\calF$ satisfies \emph{weak efficiency} if for each $W_1, W_2 \in \pow_k(C)$ and $A\in\calA(C,V)$ where no voter approves a candidate in $W_2\setminus W_1$, it holds that $W_1 \succeq_{\calF(A)} W_2$.
}

For $k=1$, i.e., in the single-winner setting, we see that weak efficiency reduces to the following statement: if no voter approves a candidate~$d$, then any other candidate is at least as preferable as~$d$.

The final axiom, continuity~\citep{young74,you:j:scoring-functions} (also known in the literature as the Ar\-chi\-medean property~\citep{smi:j:scoring-rules} or the Overwhelming Majority axiom~\citep{Myerson1995}), describes the influence of large majorities in the process of making a decision. Continuity enforces that a large enough group of voters is able to force the election of their most preferred committee. 
Continuity is pivotal in Young's characterizations of positional scoring rules~\citep{young74,you:j:scoring-functions} as it excludes specific tie-breaking mechanisms\footnote{In Young's characterization, continuity excludes \emph{composite} positional scoring rules, where one or more additional positional scoring rules are evaluated in case of ties (same score). Other tie-breaking mechanisms are already excluded by consistency and symmetry.}.

\axiom{Continuity}{
  An ABC ranking rule $\calF$ satisfies \emph{continuity} if for
  each $W_1, W_2 \in \pow_k(C)$ and $A, A'\in\calA(C)$ where $W_1 \succ_{\calF(A')} W_2$ there exists a positive integer $n$ such that 
  $W_1 \succ_{\calF(A+nA')} W_2$.
}

\subsection{A Characterization of ABC Counting Rules}\label{subsec:characterization-counting}

The following axiomatic characterization of the generic class of ABC counting rules is a powerful tool that forms the basis for further characterizations of specific ABC counting rules.
This result resembles Smith's and Young's characterization of positional scoring rules \citep{young74, smi:j:scoring-rules} as the only social welfare functions satisfying symmetry, consistency, and continuity.
Our characterization additionally requires weak efficiency, which stems from the condition that a counting function $f(x,y)$ must be weakly increasing in $x$.
If a similar condition was imposed on positional scoring rules (i.e., that positional scores are weakly decreasing), an axiom analogous to weak efficiency would be required for a characterization as well.

\begin{reptheorem}{thm:characterizationWelfareFunctions}
\thmcharacterizationWelfareFunctions
\end{reptheorem}

It is easy to verify that ABC counting rules satisfy symmetry, consistency, weak efficiency, and continuity; all this follows immediately from the definitions in Section~\ref{subsec:counting-rule-def}, in particular the summation in Equation~\eqref{eq:score}.
For example, consistency is an immediate consequence of the fact that $\score{f}(W,A+A')=\score{f}(W,A)+\score{f}(W,A')$.
Proving the other implication of Theorem~\ref{thm:characterizationWelfareFunctions}, however, requires a long and complex proof, which can be found in Appendix~\ref{sec:app:countingrules}.

Furthermore, the set of axioms used in Theorem~\ref{thm:characterizationWelfareFunctions} is minimal, i.e., any subset of axioms is not sufficient for the characterization statement to hold (see Appendix~\ref{subsec:independence}).

\section{Proportional and Disproportional ABC Counting Rules}\label{sec:proportionality}

In this section we consider axioms describing winning committees in party-list profiles and capture a specific variant of proportionality, individual excellence, or diversity.
In party-list profiles, voters and candidates are grouped into clusters, which can be intuitively viewed as political parties. 
We will show that axioms for party-list profiles are sufficient to characterize certain ABC counting rules: PAV, AV, and CC.
Using the axiomatic characterization of ABC counting rules (Theorem~\ref{thm:characterizationWelfareFunctions}), we obtain full axiomatic characterizations of these three rules.

\begin{definition}
An approval profile is a \emph{party-list profile with $p$ parties} if the set of voters can be partitioned into $N_1, N_2, \ldots, N_p$ and the set of candidates can be partitioned into $C_1, C_2, \ldots, C_p$ such that, for each $i\in[p]$, every voter in $N_i$ approves exactly $C_i$.
\end{definition}

\subsection{D'Hondt Proportionality}
\label{sec:dhondt}

In party-list profiles, we intuitively expect a proportional committee to contain as many candidates from a party as is proportional to the number of this party's supporters. There are numerous ways in which this concept can be formalized---different notions of proportionality are expressed through different methods of apportionment~\citep{BaYo82a, Puke14a}. In this section we consider one of the best known, and  most commonly used concept of proportionality: \emph{D'Hondt proportionality}~\citep{BaYo82a,Puke14a}.
\begin{wraptable}[10]{r}[0cm]{0.5\textwidth}
\begin{minipage}[h]{0.5\textwidth}
\begin{align*}
\begin{array}{c|cccc}
  & N_1  & N_2  & N_3  & N_4  \\
  \hline 
\nicefrac{|N_i|}{1} & \mathbf{9} & \mathbf{21} & \mathbf{28} & \mathbf{42} \\
\nicefrac{|N_i|}{2} & 4.5 & \mathbf{10.5} & \mathbf{14} & \mathbf{21} \\
\nicefrac{|N_i|}{3} & 3 & 7 & \mathbf{13} & \mathbf{14} \\
\nicefrac{|N_i|}{4} & 2.25 & 5.25 & 7 & \mathbf{10.5} \\
\nicefrac{|N_i|}{5} & 1.8 & 4.2 & 5.6 & 8.4
\end{array}
\end{align*}
\end{minipage}
\caption{Example for the D'Hondt method}
\label{tab:dhondt}
\end{wraptable}
The D'Hondt method is an apportionment method that works in $k$ steps. It starts with an empty committee $W = \emptyset$ and in each step it selects a candidate from that set (party) $C_i$ with a maximal value of $\frac{|N_i|}{|W \cap C_i| + 1}$; the selected candidate is added to $W$.

\begin{example} \label{ex:dHondtMethod}
Consider an election with four groups of voters, $N_1$, $N_2$, $N_3$, and $N_4$ with cardinalities equal to 9, 21, 28, and 42, respectively. Further, there are four groups of candidates $C_1 = \{c_1, \ldots, c_{10}\}$, $C_2  =\{c_{11}, \dots, c_{20}\}$, $C_3  =\{c_{21}, \ldots, c_{30}\}$, and $C_4  =\{c_{31}, \ldots, c_{40}\}$. Voters in a group $N_i$ approve exactly the candidates in $C_i$. Assume $k = 10$ and consider Table~\ref{tab:dhondt}, which illustrates the ratios used in the D'Hondt method for determining which candidate should be selected.
The 10 largest values (in bold) correspond to selected candidates.

Thus, the D'Hondt method first selects a candidate from $C_4$, next a candidate from $C_3$, next from $C_2$ or $C_4$ (their ratios in the third step are equal), etc. Eventually, in the selected committee there will be one candidate from $C_1$, two candidates from $C_2$, three from $C_3$, and four from $C_4$.
\end{example}

An important difference between the apportionment setting and our setting is that we do not necessarily assume an unrestricted number of candidates for each party. As a consequence, a party might deserve additional candidates but this is impossible to fulfill. Taking this restriction into account, we see that if the D'Hondt method picks a candidate from $C_i$ and adds it to $W$, then, for all $j$, either $\frac{|N_i|}{|W \cap C_i|} \geq \frac{|N_j|}{|W \cap C_j| + 1}$ or $C_j \subseteq W$, i.e., all candidates from party $j$ are already in the committee. Note that if $C_j \setminus W \neq \emptyset$ and $\frac{|N_i|}{|W \cap C_i|} < \frac{|N_j|}{|W \cap C_j| + 1}$, then the D'Hondt method in the previous step would rather select a candidate from $C_j$ than from $C_i$. These observations allow us to give a precise definition of D'Hondt proportional committees.

\begin{definition}\label{def:dhondt}
Let $A$ be a party-list profile with $p$ parties.
A committee $W\in \pow_k(C)$ is \emph{D'Hondt proportional for $A$} if for all $i,j\in[p]$ one of the following conditions holds:
\begin{inparaenum}[(i)]
\item $C_j \subseteq W$, or
\item $W \cap C_i = \emptyset$, or\label{cond2:def:dhondt}
\item $\frac{|N_i|}{|W \cap C_i|} \geq \frac{|N_j|}{|W \cap C_j| + 1}$.
\end{inparaenum}
\end{definition}

For the following axiom, recall that a winning committee is a maximal element in social ranking of committees, i.e., with respect to $\succeq_{\calF(A)}$.

\axiom{D'Hondt proportionality}{
An ABC ranking rule satisfies \emph{D'Hondt proportionality} if for each party-list profile $A\in\calA(C,V)$, $W\in \pow_k(C)$ is a winning committee if and only if $W$ is D'Hondt proportional for $A$.
}

Note that this axiom is weak in the sense that it only describes the expected behavior of an ABC ranking rule on party-list profiles. As we will see, however, it is sufficient to obtain an axiomatic characterization of PAV in the more general framework of ABC ranking rules.

\newcommand{\theorempavcharacterization}{Proportional Approval Voting is the only ABC counting rule that satisfies D'Hondt proportionality.}
\begin{theorem}\label{thm:pav-ranking-characterization}
\theorempavcharacterization
\end{theorem}

When we combine Theorem~\ref{thm:pav-ranking-characterization} with Theorem~\ref{thm:characterizationWelfareFunctions}, we obtain a full axiomatic characterization of Proportional Approval Voting within the class of ABC ranking rules:
PAV is the only ABC ranking rule that satisfies symmetry, consistency, continuity, and D'Hondt proportionality.
Note the absence of weak efficiency in the set of axioms that characterize PAV, since weak efficiency is implied by the other axioms (cf.~Lemma~\ref{lem:sym+con+prop->pareto} in Section~\ref{subsec:divisor}).
In Section~\ref{subsec:divisor}, we will present a generalization of Theorem~\ref{thm:pav-ranking-characterization} that applies to other apportionment methods than D'Hondt.

\subsection{Disjoint Equality}
\label{sec:disequ}

In some scenarios we might not want a multi-winner rule to be proportional.
For example, if our goal is to select a set of finalists in a contest based on a set of recommendations coming from judges or reviewers (a scenario that is often referred to as a shortlisting), candidates can be assessed independently and there is no need for proportionality. For instance, if our goal is to select 5 finalists in a contest, and if four reviewers support candidates $c_1, \ldots, c_5$ and one reviewer supports candidates $c_6, \ldots, c_{10}$ then it is very likely that we would prefer to select candidates $c_1, \ldots, c_5$ as the finalists---in contrast to what, e.g., D'Hondt proportionality suggests.

Disjoint equality is a property which might be viewed as a certain type of disproportionality. Intuitively, it requires that each approval of a candidate carries the same power: a candidate approved by a voter $v$ receives a certain level of support from $v$ which does not depend on what other candidates $v$ approves or disapproves of; in particular it does not depend on whether there are other members of a winning committee which are approved by $v$.
Disjoint equality was first proposed by Fishburn~\shortcite{fishburn78Approval} and then used by Sertel~\shortcite{sertel88Approval} as one of the distinctive axioms characterizing single-winner Approval Voting. The following axiom is its natural extension to the multi-winner setting.

\axiom{Disjoint equality}{
An ABC ranking rule $\calF$ satisfies \emph{disjoint equality} if for every profile $A \in\calA(C,V)$ 
with $\left|\bigcup_{v\in V} A(v)\right|\geq k$ and where each candidate is approved at most once, the following holds:
$W\in\pow_k(C)$ is a winning committee if and only if $W\subseteq \bigcup_{v\in V} A(v)$.
}

In other words, disjoint equality asserts that in a profile consisting of disjoint approval ballots every committee wins that consists of approved candidates.
Note that 
disjoint equality applies to an even more restricted form of party-list profiles.

\newcommand{\thmapprovalB}{Multi-Winner Approval Voting is the only ABC counting rule that satisfies disjoint equality.}
\begin{theorem}\label{thm:approval-characterizationB}
\thmapprovalB
\end{theorem}

Theorem~\ref{thm:approval-characterizationB} together with Theorem~\ref{thm:characterizationWelfareFunctions} yields an axiomatic characterization: AV is the only ABC ranking rule that satisfies symmetry, consistency, weak efficiency, continuity, and disjoint equality.

\arxiv{
It is noteworthy that the disjoint equality axiom applies to approval profiles with an arbitrary number of voters. 
This is in contrast to the original disjoint equality axiom, which has been used to axiomatically characterize single-winner Approval Voting~\citep{fishburn78Approval}: in this setting it sufficed to consider profiles with two voters. This is not the case in the multi-winner setting, as we show in Appendix~\ref{sec:appendix-disjointequ}. 
}

\subsection{Disjoint Diversity}
\label{sec:disdiv}

The disjoint diversity axiom is strongly related to the diversity principle, as it states that there exists a winning committee in which the $k$ strongest parties receive at least one seat. In other words, every party has to receive one seat before one party receives a second seat.

\axiom{Disjoint diversity}{
An ABC ranking rule $\calF$ satisfies \emph{disjoint diversity} if for every party-list profile $A \in\calA(C,V)$ with $p$ parties and $|N_1|\geq |N_2|\geq \dots \geq |N_p|$, there exists a winning committee $W$ with $W \cap C_i \neq \emptyset$ for all $i \in \{1,\dots,\min(p,k)\}$.
}

Note that disjoint diversity is a slightly weaker axiom in comparison to D'Hondt proportionality and disjoint equality since it does not characterize all winning committees for party-list profiles---it only requires the existence of one specific winning committee and does not even fully specify this committee.
As a consequence, there are several apportionment methods in the literature that imply disjoint diversity:
the Adams method, the Dean method, and the Huntington--Hill method all require that every party receives one seat before a party can obtain a second seat \citep{BaYo82a}.
Thus, it may come as a surprise that disjoint diversity nevertheless characterize a single ABC counting rule.

\newcommand{\thmccabcrulecharacterization}{The Approval Chamberlin--Courant rule is the only non-trivial ABC counting rule that satisfies disjoint diversity.}
\begin{theorem}
\thmccabcrulecharacterization\label{thm:cc_characterization}
\end{theorem}

Observe that CC does not extend the aforementioned apportionment methods because of the simple fact that it is not at all proportional. We can thus conclude that these apportionment methods do not have a counterpart in the class of ABC counting rules.
However, if we allow a tie-breaking mechanism, we find analogues.
For example, the Adams method is a divisor method similar to D'Hondt but based on the divisor sequence $(0,1,2,\dots)$.
As vote counts are first divided by 0 (defined as an arbitrarily large number), each party is guaranteed to receive one seat.
The Adams method can be extended to a ABC ranking rule: it is 
the Chamberlin--Courant rule with the $(w_1,1,\nicefrac 1 2 ,\nicefrac 1 3,\dots)$-Thiele method used to break ties between committees with the same CC score ($w_1$ is an arbitrary number).

Finally, we obtain as a corollary that CC is characterized as the only non-trivial ABC ranking rule that satisfies symmetry, consistency, weak efficiency, continuity, and disjoint diversity.

\section{Extensions}\label{sec:extensions}

In this section we discuss three extensions of our main results. First, we define a weaker form of D'Hondt proportionality, called lower quota and we show that ABC rules that satisfy lower quota must resemble PAV. Second, we extend our axiomatic characterizations of PAV and show a more general result that applies to a whole spectrum of different forms of proportionality. Third, we formulate some of our results (in particular, the characterization of PAV, CC, and the aforementioned generalization) for ABC choice rules.

\subsection{Lower Quota}
\label{sec:lower-quota}

D'Hondt proportionality determines for every party-list profile an apportionment of candidates to parties.
One may wonder if this definition of proportionality can be further weakened and still allow a characterization of PAV.
For example, the D'Hondt method is the only divisor method satisfying the lower quota axiom~\citep{BaYo82a}: intuitively, it states that a party that receives an $\alpha$ proportion of votes should receive at least $\lfloor\alpha\cdot k\rfloor$ of the $k$ available seats.
In the following we will show that this weaker axiom is not sufficient, but it characterizes ABC counting rules that are at least similar to PAV.
Let us first define lower quota for ABC ranking rules:

\axiom{Lower Quota}{
An ABC ranking rule satisfies \emph{lower quota} if for each party-list profile $A$ with $p$ parties, and a winning committee $W\in \pow_k(C)$ it holds for all $i\in\{1,\dots,p\}$ that $|W \cap C_i| \geq \left\lfloor \frac{k|N_i|}{|V|} \right\rfloor$ or $|C_i| < \left\lfloor \frac{k|N_i|}{|V|} \right\rfloor$.
}

First, let us observe that there exist ABC counting rules---other than PAV---which satisfy lower quota. 

\begin{example}\label{ex:lowerquota}
Let $m=3$ and $k=2$. Let us consider an ABC counting rule defined by the counting function $f(0, y) = 0$, and $f(1, y) = 1$ and $f(2, y) = 1.1$. This rule satisfies lower quota: Let $A$ be a party-list profile for $m = 3$ with $p \leq 3$ disjoint groups of voters $N_1, N_2, \ldots N_p$ and with their corresponding approval sets being $C_1,\dots, C_p$. For the sake of contradiction, let us assume that there exists a winning committee $W$ such that for some $i \in [p]$ we have $|C_i| \geq \big\lfloor 2 \cdot \frac{|N_i|}{|V|} \big\rfloor$ and $|W \cap C_i| < \big\lfloor 2 \cdot \frac{|N_i|}{|V|} \big\rfloor$. If $N_i = V$, then this means that a candidate who is not approved by any voter is contained in $W$, which contradicts the definition of our rule and the fact that there exist two candidates approved by some voters (since $|N_i| = |V|$, we get that $|C_i| \geq 2$). If $|N_i| < |V|$, then $\big\lfloor 2 \cdot \frac{|N_i|}{|V|} \big\rfloor$ can either be 0 or 1. Since $|W \cap C_i| < \big\lfloor 2 \cdot \frac{|N_i|}{|V|} \big\rfloor$, we conclude that $\big\lfloor 2 \cdot \frac{|N_i|}{|V|} \big\rfloor = 1$ and $|W \cap C_i|=0$. Consequently $|N_i| \geq \frac{|V|}{2}$; even if all the remaining voters from $V \setminus N_i$ approved the two members of the winning committee $W$ it is more beneficial, according to our rule, to drop one such candidate from $W$ and to add a candidate from $C_i$. Indeed, it is easy to verify that such a committee would have a higher score. This shows that our rule indeed satisfies lower quota.
\end{example}

The following shows that ABC counting rules satisfying lower quota must resemble PAV.

\newcommand{\proplowerquota}{Fix $x, y \in \naturals$ and let $m \geq y + k -x + 1$. Let $\calF$ be an ABC counting rule satisfying lower quota, and let $f$ be a counting function implementing $\calF$. It holds that:
\begin{align*}
f(x-1, y) + \frac{1}{x} \cdot f(1, 1) \cdot \frac{k-x}{k-x+1} \leq f(x, y) \leq f(x-1, y) + \frac{1}{x-1} \cdot f(1, 1) \text{.}
\end{align*}
}
\begin{proposition}\label{prop:lowerquota}
\proplowerquota
\end{proposition}

Note that $\lim_{k \to \infty} \frac{k-x}{k-x+1} = 1$, so Proposition~\ref{prop:lowerquota} asserts that---for large $k$---the value of $f(x, y)$ is roughly between $f(x-1, y) + \frac{1}{x} \cdot f(1, 1)$ and $f(x-1, y) + \frac{1}{x-1} \cdot f(1, 1)$. Recall that for PAV we have that $f(x, y) = f(x-1, y) + \frac{1}{x} \cdot f(1, 1)$ and hence Proposition~\ref{prop:lowerquota} indeed implies that an ABC counting rule satisfying lower quota must be defined by a counting function similar to PAV.

For a visualization of this result we recall Figure~\ref{fig:countingfcts} in the introduction of this paper. The gray area displays the lower and upper bound obtained from Proposition~\ref{prop:lowerquota}; to compute a lower bound we used $k=8$.

\subsection{Extension to Other Forms of Proportionality}\label{subsec:divisor}

In this section we formulate an axiom that generalizes D'Hondt proportionality. Given a sequence $\mathbf{d} = (d_1, d_2 ,\ldots)$, the $\mathbf{d}$-proportionality requires that a multi-winner rule must behave on party-list profiles as a divisor apportionment method based on the sequence of divisors $d$. Thus, for the sequence $d_{\mathrm{DHondt}} = (1, 2, 3, \ldots)$, $d_{\mathrm{DHondt}}$-proportionality is equivalent to D'Hondt proportionality.
Notably, this definition applies to other known apportionment divisor methods, such as the \emph{Sainte-Lagu\"e} (Webster) method---the divisor method based on the sequence $d_{\mathrm{SL}} = (1, 3, 5, \ldots)$. It also applies to non-linear forms of proportionality---for example, the sequence of divisors $d_{\mathrm{Penrose}} = (1, 4, 9, \ldots)$ implements the idea of square-root proportionality as devised by Penrose~\shortcite{Penr46a}, where a party should get a number of seats proportional to the square root of the number of supporters. 
In the following we use the convention that $\frac{x}{\infty} = 0$ for integers $x$. 

\begin{definition}\label{def:dprop}
Let $A$ be a non-decreasing party-list profile with $p$ parties, and $\mathbf{d} = (d_i)_{i \in \naturals}$, where $d_i \in \naturals\cup\{\infty\}$ for each $i \in \naturals$.
A committee $W\in \pow_k(C)$ is \emph{$\mathbf{d}$-proportional for $A$} if for all $i,j\in[p]$ one of the following conditions holds:
\begin{inparaenum}[(i)]
\item $C_j \subseteq W$, or
\item $W \cap C_i = \emptyset$, or\label{cond2:def:dprop}
\item $\frac{|N_i|}{d_{|W \cap C_i|}} \geq \frac{|N_j|}{d_{|W \cap C_j| + 1}}$.
\end{inparaenum}
\end{definition}

\axiom{$\mathbf{d}$-proportionality}{
Let $\mathbf{d} = (d_i)_{i \in \naturals}$ be a sequence of values from $\naturals\cup\{\infty\}$.
An ABC ranking rule satisfies \emph{$\mathbf{d}$-proportionality} if for each party-list profile $A\in\calA(C,V)$, $W\in \pow_k(C)$ is a winning committee if and only if $W$ is $\mathbf{d}$-proportional for $A$.
}

\newcommand{\thmthieleabcrulecharacterization}{Let $\mathbf{d} = (d_1, d_2, \ldots)$ be a non-decreasing sequence of values from $\naturals\cup\{\infty\}$ and let $w = (\nicefrac{1}{d_1}, \nicefrac{1}{d_2}, \ldots)$. The $w$-Thiele method is the only ABC counting rule that satisfies $\mathbf{d}$-proportionality.}
\begin{theorem}\label{thm:thiele-abc-rule-characterization}
\thmthieleabcrulecharacterization
\end{theorem}

Theorem~\ref{thm:thiele-abc-rule-characterization} contains the characterization of PAV via D'Hondt proportionality as a special case. It also gives a characterization of CC as the only ABC counting rule that is $(1,\infty, \infty, \dots)$-proportional. Note that this characterization is slightly weaker than the one via disjoint diversity (Theorem~\ref{thm:cc_characterization}), since $(1,\infty, \infty, \dots)$-proportionality specifies the behavior of the rule on all party-list profiles.
Furthermore, we can use Theorem~\ref{thm:thiele-abc-rule-characterization} to obtain axiomatic characterizations within the class of ABC ranking rules. 

\newcommand{\lemsymconproppareto}{Let $d = (d_1, d_2, \ldots)$ be a non-decreasing sequence of values from $\naturals$. An ABC ranking rule that satisfies neutrality, consistency, and $\mathbf{d}$-proportionality also satisfies weak efficiency.}
\begin{lemma}
\lemsymconproppareto\label{lem:sym+con+prop->pareto}
\end{lemma}

By combining Theorems~\ref{thm:characterizationWelfareFunctions},~\ref{thm:thiele-abc-rule-characterization}, and Lemma~\ref{lem:sym+con+prop->pareto}, we obtain the desired characterization.

\begin{corollary}\label{cor:thiele-characterization}
Let $\mathbf{d} = (d_1, d_2, \ldots)$ be a non-decreasing sequence  of values from $\naturals\cup \{\infty\}$ and let $w = (\nicefrac{1}{d_1}, \nicefrac{1}{d_2}, \ldots)$. The $w$-Thiele method is the only ABC ranking rule that satisfies symmetry, consistency, continuity, weak efficiency, and $\mathbf{d}$-proportionality. If the values from the sequence $\mathbf{d}$ do not contain $\infty$, then we do not require weak efficiency to characterize the corresponding $w$-Thiele method.
\end{corollary}

\subsection{ABC Choice Rules}\label{sec:choice_rules}

So far, we have discussed axiomatic questions concerning ABC ranking rules. We will now consider ABC choice rules, i.e., approval-based multi-winner rules that select a set of winning committees.
In particular, we will formulate some of our characterization results for choice rules.  
Let us introduce one more axiom, which is more technical and necessary for our proof technique.

\axiom{2-Nonimposition}{ An ABC choice rule $\calR$ satisfies \emph{2-Nonimposition} if for
 every two committees $W_1, W_2 \in \pow_k(C)$ there exists an approval profile $\alpha(W_1, W_2)$ such that $\calR(\alpha(W_1, W_2)) = \{W_1, W_2\}$.
}

We start by formulating an analogous result to Theorem~\ref{thm:characterizationWelfareFunctions} for ABC choice rules, but under the additional assumption of 2-Nonimposition. The proofs of the statements from this section as well as the exact axiom formulations are provided in Appendix~\ref{sec:app:choice_rules}.
  
\newcommand{\thmcharacterizationChoiceFunctions}{An ABC choice rule that satisfies 2-Nonimposition is an ABC counting rule if and only if it satisfies symmetry, consistency, weak efficiency, and continuity.}
\begin{theorem}\label{thm:characterizationChoiceFunctions}
\thmcharacterizationChoiceFunctions
\end{theorem}%

Next, we can adapt the characterization of $w$-Thiele methods to ABC choice rules

\newcommand{\lemproportionalityimpliestwononimposition}{Let $d = (d_1, d_2, \ldots)$ be a non-decreasing sequence of positive integers with $d_2 > d_1$. An ABC choice rule that satisfies consistency and $\mathbf{d}$-proportionality also satisfies 2-Nonimposition.}
\begin{lemma}\label{lem:proportionality_implies_2_nonimposition}
\lemproportionalityimpliestwononimposition
\end{lemma}

\newcommand{\thmpavabcwinnerrulecharacterization}{Let $d = (d_1, d_2, \ldots)$ be a non-decreasing sequence of positive integers with $d_2 > d_1$, and let $w = (\nicefrac{1}{d_1}, \nicefrac{1}{d_2}, \ldots)$. The $w$-Thiele method is the only ABC choice rule that satisfies symmetry, consistency, continuity and ${d}$-proportionality.}
\begin{theorem}\label{thm:pav-abc-winner-rule-characterization}
\thmpavabcwinnerrulecharacterization
\end{theorem}

In particular, the above theorem covers Proportional Approval Voting. Furthermore, we can reprove the characterization for Approval Chamberlin--Courant rule via disjoint diversity.

\newcommand{\thmccabcwinnerrulecharacterization}{The Approval Chamberlin--Courant rule is the only non-trivial ABC choice rule that satisfies symmetry, consistency, weak efficiency, continuity, and disjoint diversity.}
\begin{theorem}\label{thm:cc-abc-winner-rule-characterization}
\thmccabcwinnerrulecharacterization
\end{theorem}

Interestingly, our technique based on 2-Nonimposition cannot be used for proving an axiomatic characterization of Multi-Winner Approval Voting.

\newcommand{\propMAVfailstwononimposition}{Multi-Winner Approval Voting does not satisfy 2-Nonimposition.}
\begin{proposition}\label{prop:MAVfails2nonimposition}
\propMAVfailstwononimposition
\end{proposition}

Thus, a full characterization of
ABC counting rules within the class of ABC choice rules remains as important future work.

\section{Conclusions}
\label{sec:concl}

In this paper we analyzed a variety of different rules which all satisfy four common properties: symmetry, consistency, continuity, and weak efficiency. We identified the class of rules that is uniquely defined by these four properties: ABC counting rules.
The intuitive relevance of these four axioms is quite different: we believe that if symmetry is accepted as a prerequisite for a sensible voting rule, consistency is the essential axiom for the characterization of ABC counting rules. It can be expected that weak efficiency only guarantees that the score of a fixed voter is non-decreasing if more approved candidates are in the committee. It is a plausible assumption that voters desire to have approved candidates in the committee and hence weak efficiency only rules out pathological examples of multi-winner rules.
The role of continuity is also a technical one. We conjecture that the role of continuity in our characterization is the same as the role of continuity in (single-winner) scoring rules, i.e., removing continuity also allows for ABC counting rules that use other ABC counting rules to break ties between committees with the same score (we mention such a rule in Section~\ref{sec:disdiv}).
These arguments support our claim that ABC counting rules capture essentially the class of consistent approval-based multi-winner rules.
A formal characterization of ABC ranking rules that satisfy symmetry and consistency would be desirable to substantiate this claim and to further shed light on consistent rules.

The class of ABC counting rules is remarkably broad and includes rules such as Proportional Approval Voting, Approval Chamberlin--Courant and Multi-Winner Approval Voting, for all of which we have provided an axiomatic characterization.
These characterizations are obtained by axioms that describe desirable outcomes for certain simple profiles, so-called party-list profiles.
This is a fruitful approach as it is much easier to formally define concepts such as proportionality or diversity on these simple profiles.
In such profiles it is also easy to formulate properties which quantify trade-offs between individual efficiency, proportionality, and diversity.
Furthermore, the simpler domain of party-list profiles is sufficient to explain the difference between the rules: PAV, AV, and CC can be obtained by extending three different principles defined for party-list profiles to the more general domain by additionally imposing the same set of axioms. Therefore, their defining differences can be found in party-list profiles.

Our results are general and extend to other concepts definable on party-list profiles, e.g., Sainte-Lagu\"e~\citep{Puke14a,BaYo82a} or square-root proportionality~\citep{Penr46a,SlZy06a}. 
Square-root proportionality follows the \emph{degressive proportionality} principle~\citep{RePEc:ucp:jpolec:doi:10.1086/670380}, which suggests that smaller populations should be allocated more representatives than linear proportionality would require.
This can be achieved by using a flatter counting function than $f_\text{PAV}$ and by that we obtain rules which increasingly promote diversity within the committee.
An extreme example is the Approval Chamberlin--Courant rule, where the diversity within a winning committee is strongly favored over proportionality.
On the other hand, using steeper counting functions results in rules with a utilitarian focus, i.e., rules that tend to select candidates with the high total support from voters, and ignoring issues of proportionality of representation. Multi-Winner Approval Voting is an extreme example of a rule which does not guarantee virtually any level of proportionality.

\section*{Acknowledgments}

We thank Edith Elkind, Piotr Faliszewski, Paul Harrenstein and Dominik Peters for insightful discussions. We are grateful to Marie-Louise Lackner for providing feedback on early drafts of this paper. 
Furthermore, we thank Itzhak Gilboa for his supportive comments on our paper.
Martin Lackner was supported by the European Research Council (ERC) under grant number 639945 (ACCORD) and the Austrian Science Foundation FWF, grant P31890. 
Piotr Skowron was supported by the ERC grant 639945 (ACCORD), and by the Foundation for Polish Science within the Homing programme (Project title: "Normative Comparison of Multiwinner Election Rules").

\appendix
\section{Characterization of ABC Counting Rules}
\label{sec:app:countingrules}

In this section we prove the main technical result of this paper:

\begin{reptheorem}{thm:characterizationWelfareFunctions}
\thmcharacterizationWelfareFunctions
\end{reptheorem}

The following definitions and notation will prove useful:

\begin{itemize}
\item For each $\ell \in [0,m]$, we say that an approval profile $A$ is \emph{$\ell$-regular} if each voter in $A$ approves of exactly $\ell$ candidates.
\item We say that $A$ is \emph{$\ell$-bounded} if each voter in $A$ approves of at most $\ell$ candidates.
\item We write $\set(A)$ to denote $\{A(v): v\in V\}$ and by that ignore multiplicities of votes.
\item Sometimes we associate an approval set $S \subseteq C$ with the single-voter profile $A\in\calA(C,\{1\})$ such that $A(1)=S$; in such a  case we write $\calF(S)$ as a short form of $\calF(A)$ for appropriately defined $A$.
\end{itemize}

\paragraph{Committee scoring rules.}

Before we start describing our construction, let us recall the definition of committee scoring rules~\citep{skowron2019axiomatic}, a concept that will play an instrumental role in our further discussion.
Linear order-based committee (LOC) ranking rules, in contrast to ABC ranking rules, assume that voters' preferences are given as linear orders over the set of candidates.
For a finite set of voters $V=\{v_1,\dots,v_n\}\subset \naturals$, a \emph{profile of linear orders over $V$}, $P= (P(v_1),\dots,P(v_n))$, is an $n$-tuple of linear orders over $C$ indexed by the elements of $V$, i.e., for all $v\in V$ we have $P(v)\in\mathscr L(C)$.
A \emph{linear order-based committee ranking rule (LOC ranking rule)} is a function that maps profiles of linear orders to $\mathscr W(\pow_k(C))$, the set of weak orders over committees.

Let $P$ be a profile of linear orders over $V$.
For a vote $v$ and a candidate $a$, by $\pos_{v}(a, P)$ we denote
the position of $a$ in $P(v)$ (the top-ranked candidate has position
$1$ and the bottom-ranked candidate has position $m$). %
For a vote $v\in V$ and a committee $W \in \pow_k(C)$, we write $\pos_v(W, P)$ to denote the 
set of positions of all members of $W$ in ranking $P(v)$, i.e.,
$\pos_v(W, P) = \{\pos_v(a, P): a \in W\}$. 
A \emph{committee scoring function} is a mapping $\posf\colon \pow_k([m]) \to \reals$ that for each possible position that a committee can occupy in a ranking (there are ${m \choose k}$ of all possible positions), assigns a score. Intuitively, for each $I \in \pow_k([m])$ value $\posf(I)$ can be viewed as the score assigned by a voter $v$ to the committee whose members stand in $v$'s ranking on positions from set $I$.
Additionally, a committee scoring function $\posf(I)$ is required to satisfy weak dominance, which is defined as follows.
Let $I,J\in \pow_k([m])$
such that $I
= \{i_1, \ldots, i_k\}$, $J = \{j_1, \ldots, j_k\}$, and that
$i_1 < \cdots < i_k$ and $j_1 < \cdots < j_k$. We say that $I$
dominates $J$ if for each $t \in [k]$ we have $i_t \leq j_t$. Weak dominance holds if $I$ dominating $J$ implies that $\posf(I) \geq \posf(J)$.

For a profile of linear orders $P$ over $C$ and a committee $W\in\pow_k(C)$,
we write $\score{f}(W, P)$ to denote the total score that the
voters from $V$ assign to committee $W$. Formally, we have that
$\score{\posf}(W, P) = \sum_{v \in V} \posf(\pos_{v}(W, P))$.
An LOC ranking rule $\calG$ is an \emph{LOC scoring rule} if
there exists a committee scoring function $\posf$ such that for each
$W_1, W_2\in \pow_k(C)$ and profile of linear orders $P$ over $V$, we have that $W_1$ is strictly preferred to $W_2$ with respect to the weak order $\calG(P)$
if and only if $\score{\posf}(W_1, P) > \score{\posf}(W_2, P)$.

The axioms from Section~\ref{subsec:basic-axioms} can be naturally formulated for LOC ranking rules. We will use these formulations of the axioms in the proof of Lemma~\ref{lemma:ordinalFunctionPreservesProperties}. For the sake of readability we do not recall their definitions here, but rather in the proof, where they are used.

\paragraph{Overview of the proof of Theorem~\ref{thm:characterizationWelfareFunctions}.} 
As mentioned before, it is easy to see that ABC counting rules satisfy symmetry, consistency, weak efficiency, and continuity.
The proof of the other direction consists of several steps.

We start in Section~\ref{subsec:eff} by proving that weak efficiency in conjunction with the other axioms implies a stronger efficiency axiom, which proves useful in the subsequent constructions.
In Section~\ref{subsec:l-regular}, we prove that the characterization theorem holds for the very restricted class of $\ell$-regular profiles, i.e., profiles where every voter approves exactly $\ell$ candidates.
To this end, we construct a collection of LOC rules $\left\{\calG_{\ell}\right\}_{\ell=1\dots m}$ based on how $\calF$ operates on $\ell$-regular profiles.
We then show that the LOC ranking rule $\calG_{\ell}$ satisfies equivalent axioms to symmetry, consistency, weak efficiency, and continuity.
This allows us to apply a theorem by Skowron, Faliszewski, and Slinko~\citep{skowron2019axiomatic}, who proved that LOC ranking rules satisfying these axioms are in fact LOC scoring rules.
Thus, there exists a corresponding committee scoring function $g_\ell$, which in turn defines a counting function $f_\ell$.
As a last step, we show that $f_\ell$ implements $\calF$ on $\ell$-regular approval profiles and thus prove that Theorem~\ref{thm:approval-characterizationB} holds restricted to $\ell$-regular approval profiles.

In Section~\ref{subsec:arbitrary}, we extend this restricted result to arbitrary approval profiles. For each $\ell \in [m]$ we have obtained a counting function $f_{\ell}$ which defines $\calF$ on $\ell$-regular profiles.
Our goal is to show that there exists a linear combination of these counting functions $f = \sum_{\ell \in [m]}\gamma_{\ell} f_{\ell}$ which defines $\calF$ on arbitrary profiles. We define the corresponding coefficients $\gamma_{1}, \dots, \gamma_m$ inductively. We first construct two specific committees $W_1^*$ and $W_2^*$, which we use to scale the coefficients, and additionally, in order to define coefficient $\gamma_{\ell+1}$ we construct two specific votes, $a_{\ell+1}^*$ and $b_{\ell+1}^*$, with exactly $\ell+1$, and at most $\ell$ approved candidates, respectively. We define coefficient $\gamma_{\ell+1}$ using the definition of $f$ for $\ell$-bounded profiles and by exploring how $\calF$ compares committees $W_1^*$ and $W_2^*$ for very specific profiles which are build from certain numbers of votes $a_{\ell+1}^*$ and $b_{\ell+1}^*$. This concludes the construction of $f$.

Showing that $f = \sum_{\ell \in [m]}\gamma_{\ell} f_{\ell}$ implements $\calF$ requires a rather involved analysis, which is divided into several lemmas. In Lemma~\ref{lemma:specificVotes} we show that $f$ implements $\calF$, but only for the case when $\calF$ is used to compare $W_1^*$ and $W_2^*$, and only for very specific profiles.
In Lemma~\ref{lemma:WqW2Fixed} we still assume that $\calF$ is used to compare only $W_1^*$ and $W_2^*$, but this time we extend the statement to arbitrary profiles. In Lemma~\ref{lemma:arbitraryCommittees} we show the case when $\calF$ is used to compare $W_1^*$ with any other committee. We complete this reasoning with a short discussion explaining the validity of our statement in its full generality. Each of the aforementioned lemmas is based on a different idea and they build upon each other. The main proof technique is to transform simple approval profiles to more complex ones and argue that certain properties are preserved due to the required axioms.

\subsection{A Stronger Efficiency Axiom}
\label{subsec:eff}

In the subsequent proofs we will use the following \emph{strong efficiency} axiom:

\axiom{Strong Efficiency}{
  An ABC ranking rule $\calF$ satisfies \emph{strong efficiency}
  if for $W_1, W_2 \in \pow_k(C)$ and $A\in\calA(C,V)$ where for every voter $v\in V$ we have $|A(v) \cap W_1| \geq |A(v) \cap W_2|$, it holds that $W_1 \succeq_{\calF(A)} W_2$.
}

For $k=1$, i.e., in the single-winner setting, strong efficiency is the well-known Pareto efficiency axiom, which requires that if a candidate $c$ is unanimously preferred to candidate $d$, then $d \succeq c$ in the collective ranking~\citep{moulinAxioms}.

The following lemma shows that strong efficiency in the context of neutral and consistent rules is implied by its weaker counterpart. 

\newcommand{\lemconsparetoefficiency}{An ABC ranking rule that satisfies neutrality, consistency and weak efficiency also satisfies strong efficiency.}
\begin{lemma}\label{lem:cons+pareto->efficiency}
\lemconsparetoefficiency
\end{lemma}

\begin{proof}
Let $\calF$ be an ABC ranking rule that satisfies neutrality, consistency and weak efficiency.
Further, let $W_1, W_2 \in \pow_k(C)$ and $A\in\calA(C,V)$ such that for every vote $v\in V$ we have $|A(v) \cap W_1| \geq |A(v) \cap W_2|$.
We have to show that $W_1 \succeq_{\calF(A)} W_2$. %
Fix $v\in V$ and let $A_v\in \calA(C,\{1\})$ be the profile containing the single vote $A(v)$.
Now, let us consider a committee $W_2'$ constructed from $W_2$ in the following way. We obtain $W_2'$
from $W_2$ by replacing candidates in $W_2\setminus A(v)$ with candidates from $A(v)$ so that
$|A(v) \cap W_2'| = |A(v) \cap W_1|$.
Note that $A(v) \cap W_2 \subseteq A(v) \cap W_2'$ and hence candidates in $A(v) \cap (W_2\setminus W_2') = \emptyset$.
Hence by weak efficiency we get that $W_2'\succeq_{\calF(A_v)} W_2$.
Furthermore, neutrality implies that $W_2'\sim_{\calF(A_v)} W_1$ and by transitivity we infer that  $W_1 \succeq_{\calF(A_v)} W_2$.
The final step is to apply consistency. For every $v\in V$, $W_1\succeq_{\calF(A_v)} W_2$.
Hence also for their combination $\sum_{v\in V} A_v= A$ we have $W_1\succeq_{\calF(A)} W_2$.
\end{proof}

\subsection{$\calF$ is an ABC Counting Rule on $\ell$-Regular Approval Profiles}
\label{subsec:l-regular}

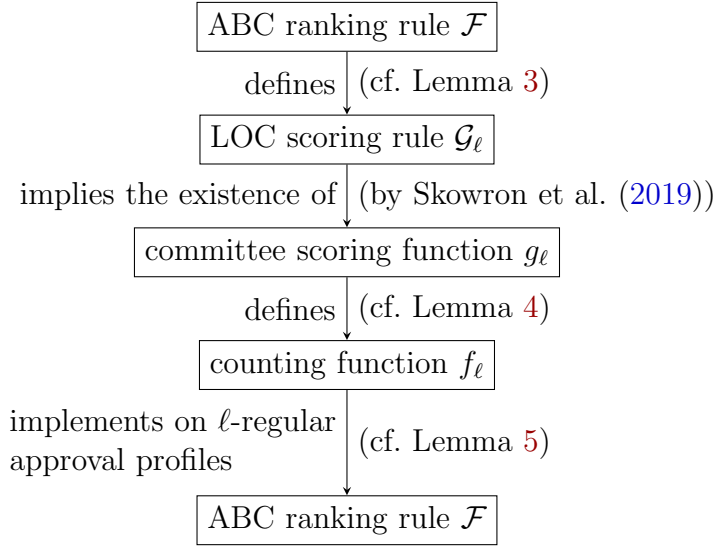
\begin{figure}
\begin{center}
\begin{tikzpicture}[level 4/.style={level distance=5em},sibling distance=10em,
  every node/.style = {draw, align=center}]]
  \node {ABC ranking rule $\calF$}
    child { node (B) {LOC scoring rule $\calG_\ell$}
     child { node (C) {committee scoring function $g_\ell$} 
       child { node {counting function $f_\ell$}
       child { node {ABC ranking rule $\calF$}
      edge from parent[-stealth] node[left,draw=none,align=left] {implements on $\ell$-regular\\approval profiles} node[right,draw=none] {(cf.\ Lemma~\ref{lem:implements-on-approval-profiles})}
    } edge from parent[-stealth] node[left,draw=none] {defines} node[right,draw=none] {(cf.\ Lemma~\ref{lem:g ell=f})} 
    } edge from parent[-stealth] node[left,draw=none] {implies the existence of} node[right,draw=none] {(by \citet{skowron2019axiomatic})} 
    } edge from parent[-stealth] node[left,draw=none] {defines} node[right,draw=none] {(cf.\ Lemma~\ref{lemma:ordinalFunctionPreservesProperties})} 
    } ;
\end{tikzpicture}
\end{center}
\caption{A diagram illustrating the reasoning used in Section~\ref{subsec:l-regular} to prove that in $\ell$-regular approval profiles, $\calF$ is a counting rule.}
\label{fig:l-regular-proof-overview}
\end{figure}

Recall that we assume that $\calF$ is an ABC ranking rule satisfying symmetry, consistency, weak efficiency, and continuity.
If $\calF$ is trivial, i.e., if $\calF$ always maps to the trivial relation, then $\calF$ is the counting rule implemented by $f(x,y)=0$. Thus, hereinafter we assume that $\calF$ is a fixed, non-trivial ABC ranking rule satisfying anonymity, neutrality, weak efficiency, and continuity. By Lemma~\ref{lem:cons+pareto->efficiency} we can also assume that $\calF$ satisfies strong efficiency.

As a first step, we will prove in this section that $\calF$ restricted to $\ell$-regular approval profiles is an ABC counting rule, i.e., that there exists a counting function that implements $\calF$ on $\ell$-regular approval profiles.
For an overview of the argument we refer the reader to Figure~\ref{fig:l-regular-proof-overview}.

For each $\ell \in [m]$, from $\calF$ we construct an LOC ranking rule, $\calG_{\ell}$, as follows. For a profile  of linear orders $P$, by $\toapproval(P, \ell)$ we denote the approval preference profile where voters approve of their top $\ell$ candidates. 
We define for every $\ell\in[m]$ an LOC ranking rule $\calG_{\ell}$, as:
\begin{align}\label{eq:calGDefinition}
\calG_{\ell}(P) = \calF\big(\toapproval(P, \ell)\big).%
\end{align}
Lemma~\ref{lemma:ordinalFunctionPreservesProperties}, below, shows that our construction preserves the axioms under consideration and consequently that $\calG_\ell$ is an LOC scoring rule. As mentioned before, this lemma heavily builds upon a result of Skowron, Faliszewski, and Slinko~\citep{skowron2019axiomatic}.

\begin{lemma}\label{lemma:ordinalFunctionPreservesProperties}
Let $\calF$ be an ABC ranking rule satisfying symmetry, consistency, strong efficiency and continuity. Then for each $\ell \in [m]$, the LOC ranking rule $\calG_{\ell}$ defined by Equation~\eqref{eq:calGDefinition} is an LOC scoring rule.\label{lem:apply-skowron-et-al}
\end{lemma}
\begin{proof}
The proof of this lemma relies on the main theorem of Skowron, Faliszewski, and Slinko~\citep{skowron2019axiomatic}:
an LOC ranking rule is a LOC scoring rule if and only if it satisfies anonymity, neutrality, consistency, committee dominance, and continuity.
We thus have to verify that $\calG_{\ell}$ satisfies these axioms.
Note that since $\calG_{\ell}$ is an LOC ranking rule, the corresponding axioms differ slightly from the ones introduced in Section~\ref{subsec:basic-axioms}.
Thus, in the following we introduce each of these axioms for LOC ranking rules and prove that it is satisfied by $\calG_\ell$ for arbitrary $\ell$.

(Anonymity) An LOC ranking rule $\calG$ satisfies anonymity if for
  every two sets of voters $V, V' \subseteq
  \naturals$ such that $|V| = |V'|$, for each bijection $\rho: V \to V'$ and
  for every two preference profiles $P_1 \in \calP(C,V)$ and $P_2 \in
  \calP(C,V')$ such that $P_1(v) = P_2(\rho(v))$ for each $v \in V$,
  it holds that $\calG(P_1) = \calG(P_2)$.
  Let $V,V',\rho, P_1, P_2$ be defined accordingly.
  Note that $P_1(v) = P_2(\rho(v))$ implies that $\toapproval(P_1, \ell)(v) = \toapproval(P_2, \ell)(\rho(v))$.
  Hence, by anonymity of $\calF$, \[\calG(P_1)=\calF\big(\toapproval(P_1, \ell)\big)=\calF\big(\toapproval(P_2, \ell)\big)=\calG(P_2).\]

(Neutrality) An LOC ranking rule $\calG$ satisfies neutrality if for each permutation
  $\sigma$ of $A$ and every two preference profiles $P_1, P_2$
  over the same voter set $V$ with $P_1 = \sigma(P_2)$, it holds that $\calG(P_1) = \sigma( \calG(P_2) )$.
  Let $P_1,P_2,V$, and $\sigma$ be defined accordingly.
  Note that $\toapproval(P_1, \ell) = \sigma(\toapproval(P_2, \ell))$.
  Then, by neutrality of $\calF$, \[\calG(P_1) = \calF\big(\toapproval(P_1, \ell)\big)= \calF(\sigma\big(\toapproval(P_2, \ell)\big)) =\sigma(\calF\big(\toapproval(P_2, \ell)\big)) = \sigma( \calG(P_2) ).\]

(Consistency) An LOC ranking rule $\calG$ satisfies consistency if for every two
  profiles $P_1$ and $P_2$ over disjoint sets of voters, $V \subset \naturals$
  and $V' \subset \naturals$, $V \cap V' = \emptyset$, and every two committees
  $W_1, W_2 \in \pow_k(C)$, (i) if $W_1 \succ_{\calG(P_1)} W_2$ and $W_1
  \succeq_{\calG(P_2)} W_2$, then it holds that $W_1 \succ_{\calG(P_1+P_2)} W_2$ and (ii) if $W_1 \succeq_{\calG(P_1)} W_2$ and $W_1
  \succeq_{\calG(P_2)} W_2$, then it holds that $W_1 \succeq_{\calG(P_1+P_2)} W_2$.
  Let $P_1,P_2,V,V',W_1$, and $W_2$ be defined accordingly.
  Let us prove~(i).
  If $W_1 \succ_{\calG(P_1)} W_2$, then $W_1 \succ_{\calF(\toapproval(P_1, \ell))} W_2$.
  Analogously, if $W_1 \succeq_{\calG(P_2)} W_2$, then $W_1 \succeq_{\calF(\toapproval(P_2, \ell))} W_2$.
  By consistency of $\calF$, we know that $W_1 \succ_{\calF(\toapproval(P_1, \ell)+\toapproval(P_2, \ell))} W_2$.
  Clearly, $\toapproval(P_1, \ell)+\toapproval(P_2, \ell)=\toapproval(P_1+P_2, \ell)$.
  We can conclude that $W_1 \succ_{\calF(\toapproval(P_1+P_2, \ell))} W_2$ and hence $W_1 \succ_{\calG(P_1+P_2)} W_2$.
  The proof of (ii) is analogous.
  
(Committee dominance) An LOC ranking rule $\calG$ satisfies committee dominance if for every two committees $W_1, W_2 \in \pow_k(C)$ and each
  profile $P\in \calP(C,V)$ where for every vote $v \in V$, $\pos_{v}(W_1)$
  dominates $\pos_{v}(W_2)$, it holds that $W_1 \succeq_{\calG(P)} W_2$.
  Let $W_1,W_2$, and $P$ be defined accordingly.
  If $\pos_{v}(W_1)$ dominates $\pos_{v}(W_2)$, then clearly for each $v \in V$, $|\toapproval(P, \ell)(v)\cap W_1|\geq |\toapproval(P, \ell)(v)\cap W_2|$.
  By strong efficiency of $\calF$, $W_1 \succeq_{\calG(P)} W_2$.
  
(Continuity) An LOC ranking rule $\calG$ satisfies continuity if for
  every two committees $W_1, W_2 \in \pow_k(C)$ and every two profiles $P_1$
  and $P_2$ where $W_1 \succ_{\calG(P_2)} W_2$, there exists a number $n \in
  \naturals$ such that $W_1 \succ_{\calG(P_1+nP_2)} W_2$. This is an immediate consequence of the fact that $\calF$ satisfies continuity.
\end{proof}

Lemma~\ref{lemma:ordinalFunctionPreservesProperties} shows that there exists a committee scoring function implementing rule $\calG_{\ell}$.
The following lemma shows that this committee scoring function has a special form that allows it to be represented by a counting function.

\begin{lemma}
For $\ell\in[m]$, let $g_{\ell}:\pow_k([m])\to \reals$ be a committee scoring function that implements $\calG_{\ell}$. There exists a counting function $f_\ell$ such that that:\label{lem:g ell=f}
\begin{align*}
g_{\ell}(I) = f_\ell(|\{i \in I \colon i \leq \ell\}|,\ell) \qquad \text{for each $I\in [m]_k$ and $\ell\in[m]$.}
\end{align*}
\end{lemma}

\begin{proof}
We have to show that for an arbitrary profile of linear orders $P$ over $V$ and some $v\in V$, two committees $W_1$ and $W_2$ have the same score $g_{\ell}(\pos_v(W_1))=g_{\ell}(\pos_v(W_2))$ given that \[|\{i \in \pos_v(W_1) \colon i \leq \ell\}|=|\{i \in \pos_v(W_2) \colon i \leq \ell\}|.\]
From the neutrality of $\calF$, we see that if $v$ has the same number of approved members in $W_1$ as in $W_2$, $W_1$ and $W_2$ are equally good with respect to $\calF$. Thus if $W_1$ and $W_2$ have the same number of members in the top $\ell$ positions in $v$, then $W_1$ and $W_2$ are also equally good with respect to $\calG_{\ell}$. Hence the scores assigned by $g_\ell$ to the positions occupied by $W_1$ and $W_2$ are the same.
\end{proof}

We are now ready to prove Lemma~\ref{lemma:on_l_regular_profiles}, which provides the main technical conclusion of this section.

\begin{lemma}\label{lemma:on_l_regular_profiles}
For each $\ell \in [m]$, the counting function $f_\ell(a,\ell)$, as defined in the statement of Lemma~\ref{lem:g ell=f}, implements $\calF$ on $\ell$-regular approval profiles.\label{lem:implements-on-approval-profiles}
\end{lemma}
\begin{proof}
For each $\ell$-regular approval profile $A$ we can create an ordinal profile $\toordinal(A,\ell)$ where voters put all approved candidates in their top $\ell$ positions (in some fixed arbitrary order) and in the next $(m -\ell)$ positions the candidates that they disapprove of (also in some fixed arbitrary order). Naturally, $\toapproval(\toordinal(A,\ell), \ell) = A$. 
Thus, a committee $W_1$ is preferred over $W_2$ in $A$ according to $\calF$ if and only if $W_1$ is preferred over $W_2$ in $\toordinal(A, \ell)$ according to $\calG_{\ell}$. Since $\calG_{\ell}$ is an LOC scoring rule, the previous statement holds if and only if $W_1$ has a higher score than $W_2$ according to the committee scoring function $g_{\ell}$. This is equivalent to $W_1$ having a higher score according to $f_{\ell}$ (Lemma~\ref{lem:g ell=f}).
We conclude that $W_1$ is preferred over $W_2$ in $A$ according to $\calF$ if and only $W_1$ has a higher score according to $f_\ell$.
Consequently, we have shown that $\calF$ is an ABC counting rule for $\ell$-regular approval profiles.
\end{proof}
As the construction in the proof of Lemma~\ref{lemma:on_l_regular_profiles} relies on $\toordinal(A,\ell)$ and so it applies only to profiles where each voter approves the same number of candidates, we need new ideas to prove that $\calF$ is an ABC counting rule on arbitrary profiles. We explain these ideas in the following section.

\subsection{$\calF$ is an ABC Counting Rule on Arbitrary Profiles}
\label{subsec:arbitrary}

We now generalize the result of Lemma~\ref{lem:implements-on-approval-profiles} for $\ell$-regular profiles to arbitrary approval profiles.
We will use here the following notation.

\begin{definition}
For an approval profile $A\in\calA(C,V)$ and $x \in [0,m]$ we write $\most(A, \ell)$ to denote the profile consisting of all votes $v\in V$ with $A(v)\leq \ell$, i.e., $\most(A, \ell)\in\calA(C,V')$ with $V'=\{v\in V\colon A(v)\leq \ell\}$ and $\most(A, \ell)(v)=A(v)$ for all $v\in V'$.
Analogously, we write $\exact(A, \ell)$ to denote the profile consisting of all votes $A(v)$, for $v\in V$ with $A(v)= \ell$.
\end{definition}
\noindent
Clearly, $\most(A, \ell)$ is $\ell$-bounded and $\exact(A', \ell)$ is $\ell$-regular.

Now, let $\{f_{\ell}\}_{\ell \leq m}$ be the family of counting functions witnessing that $\calF$, when applied to $\ell$-regular profiles, is an ABC counting rule (cf.\ Lemma~\ref{lemma:on_l_regular_profiles}). From $\{f_{\ell}\}_{\ell \leq m}$ we will now construct a single counting function $f$ that witnesses that $\calF$ is an ABC counting rule. 
Since $f$ and $f_\ell$ have to produce the same output on $\ell$-regular profiles, it would be tempting to define $f(x,\ell)=f_\ell(x,\ell)$.
However, this simple construction does not work. Instead, we will find constants $\gamma_1,\dots,\gamma_m$ such that $f(x,\ell)=\gamma_\ell \cdot f_\ell(x,\ell)$ and show that with this construction we indeed obtain a counting function implementing $\calF$.

For this construction, let us fix two arbitrary committees $W_1^*$, $W_2^*$ with the smallest possible size of the intersection. In particular, $W_1^* \cap W_2^* = \emptyset$ for $m \geq 2k$.
Let $W_1^* \setminus W_2^* = \{a_1, \ldots a_t\}$, and let $W_2^* \setminus W_1^* = \{b_1, \ldots b_t\}$.
By $\sigma^*$ we denote the permutation that swaps $a_1$ with $b_1$, $a_2$ with $b_2$, etc., and that is the identity elsewhere.

We will define $\gamma_1,\dots,\gamma_m$ inductively. 
For the base case we set $f(0,0)=0$. 
Now, let us assume that $f$ is defined on $[0,k]\times[0,\ell]$ and that $f$ implements $\calF$ on $\ell$-bounded profiles.
To choose $\gamma_{\ell+1}$, we distinguish the following three cases:

\begin{description}
\item[Case (A).] If in all $(\ell+1)$-regular profiles $A$ it holds that $W_1^* \sim_{\calF(A)} W_2^*$, then 
	we set $\gamma_{\ell+1}=0$.

\item[Case (B).] If we are not in Case (A) and in all $\ell$-bounded profiles $A$ it holds that $W_1^* \sim_{\calF(A)} W_2^*$, then
	we set $\gamma_{\ell+1}=1$.

\item[Case (C).] Otherwise, there exist a single-vote $(\ell+1)$-regular profile $A$ such that $W_1^* \neq_{\calF(A)} W_2^*$ and a single-vote $\ell$-bounded profile $A'$ such that $W_1^* \neq_{\calF(A')} W_2^*$.
Indeed, if for all $(\ell+1)$-regular single-vote profiles $A\in\calA(C,\{1\})$ it holds that $W_1^* \sim_{\calF(A)} W_2^*$, then by consistency this holds for all $(\ell+1)$-regular profiles, which is a precondition of Case~(A). Similarly, if for all $\ell$-bounded single-vote profiles $A\in\calA(C,\{1\})$ it holds that $W_1^* \sim_{\calF(A)} W_2^*$, then by consistency this holds for all $\ell$-bounded profiles (precondition of Case~(B)). Consequently, the profiles $A$ and $A'$ can be chosen to consist of single votes.

In the following, by slight abuse of notation, we identify a set of approved candidates with its corresponding single-vote profile.
Let $a_{\ell+1}^*\subseteq C$ be a vote such that
\begin{inparaenum}[(i)]
\item $|a_{\ell+1}^*|= \ell+1$,
\item $W_1^* \succ_{\calF(a_{\ell+1}^*)} W_2^*$, and
\item such that the difference between the scores of $W_1^*$ and $W_2^*$ is maximized.
\end{inparaenum}
Furthermore, let $b_{\ell+1}^*\subseteq C$ be a vote such that
\begin{inparaenum}[(i)]
\item $|b_{\ell+1}^*|\leq \ell$,
\item $W_1^* \succ_{\calF(b_{\ell+1}^*)} W_2^*$, and
\item such that the difference between the scores of $W_1^*$ and $W_2^*$ is maximized.
\end{inparaenum}
For each $x, y \in \naturals$ we define the profile $S(x, y)$ as:
\begin{align*}
S(x,y) = x\cdot \sigma^*(a_{\ell+1}^*) + y\cdot b_{\ell+1}^* \textrm{.}
\end{align*}
Let us define $t^*_{\ell+1}$ as:
\begin{align}
t^*_{\ell+1} = \sup \Big\{ \frac{x}{y} \colon W_1^* \succ_{S(x, y)} W_2^* \Big\} \textrm{,}\label{eq:def:tstar}
\end{align}
which is a well-defined positive real number as we show in Lemma~\ref{lemma:threshold}.
We define:
\begin{align*}
\gamma_{\ell+1} = \frac{\score{f}(W_1^*, b_{\ell+1}^*) - \score{f}(W_2^*, b_{\ell+1}^*)}{t^*_{\ell+1} \cdot \Big(\score{f_{\ell+1}}(W_1^*, a_{\ell+1}^*) - \score{f_{\ell+1}}(W_2^*, a_{\ell+1}^*)\Big)} \textrm{.}
\end{align*}
\end{description}

This concludes the construction of $f$. Let us now show that $t^*_{\ell+1}$ is a positive real number and that it defines a threshold:

\begin{lemma}\label{lemma:threshold}
The supremum $t^*_{\ell+1}$, as defined by Equation~\eqref{eq:def:tstar}, is a positive real number.
Furthermore, if $\nicefrac{x}{y} < t^*_{\ell+1}$, then $W_1^* \succ_{S(x, y)} W_2^*$. If $\nicefrac{x}{y} > t^*_{\ell+1}$, then $W_2^* \succ_{S(x, y)} W_1^*$.
\end{lemma}
\begin{proof}
Let us argue that $t^*_{\ell+1}$ is well defined. By continuity there exists $y$ such that $W_1^* \succ_{S(1, y)} W_2^*$. Consequently, the set in \eqref{eq:def:tstar} is nonempty. Also by continuity, there exists $x$ such that $W_2^* \succ_{S(x, 1)} W_1^*$. Further, we observe that for each $x'$, $y'$ with $\nicefrac{x'}{y'} > x$ it also holds that $W_2^* \succ_{S(x', y')} W_1^*$. Indeed, since $S(x', y') = S(xy', y') + S(x' - xy', 0)$, we infer that in such case $S(x', y')$ can be split into $y'$ copies of $S(x, 1)$ and $x' - xy'$ copies of $\sigma^*(a_{\ell+1}^*)$. By consistency we get $W_2^* \succ_{S(x', y')} W_1^*$. Thus, the set in \eqref{eq:def:tstar} is bounded, and so $t^*_{\ell+1}$ is a positive real number.

To show the second statement, let us assume that $\nicefrac{x}{y} < t^*_{\ell+1}$. From the definition of $t^*_{\ell+1}$ we infer that there exist $x',y' \in \naturals$, such that $\nicefrac{x}{y} < \nicefrac{x'}{y'}$ and such that $W_1^* \succ_{S(x', y')} W_2^*$. By consistency, it also holds that $W_1^* \succ_{S(xx', xy')} W_2^*$.
Since $W_1^* \succ_{S(0, 1)} W_2^*$ and $x'y - xy' > 0$ and we get that $W_1^* \succ_{S(0, x'y - xy')} W_2^*$. Now, observe that
\begin{align*}
S(xx', x'y) = S(xx', xy') + S(0, x'y - xy').
\end{align*}
Thus, from consistency infer that 
$W_1^* \succ_{S(xx', x'y)} W_2^*$. Again, by consistency we get that $W_1^* \succ_{S(x, y)} W_2^*$.

Next, let us assume that $\nicefrac{x}{y} > t^*_{\ell+1}$. Then, there exist $x',y' \in \naturals$, such that $\nicefrac{x}{y} > \nicefrac{x'}{y'}$ and such that $W_2^* \succ_{S(x', y')} W_1^*$. Similarly as before, we get that $W_2^* \succ_{S(x'y, yy')} W_1^*$ and since $xy' - x'y > 0$ we get that $W_2^* \succ_{S(xy' - x'y, 0)} W_1^*$. Since $S(xy', yy') = S(x'y, yy') + S(xy' - x'y, 0)$,
consistency implies that
$W_2^* \succ_{S(xy', yy')} W_1^*$. Finally, we get that $W_2^* \succ_{S(x, y)} W_1^*$, which completes the proof.
\end{proof}

In the remainder of this section, we prove that $f$ is indeed a counting function that implements $\calF$ and thus $\calF$ is an ABC counting rule.
We prove this for increasingly general profiles, starting with very simple ones, and at first we prove a slightly weaker relation between $f$ and $\calF$.

\begin{lemma}\label{lemma:specificVotes}
Let us fix $\ell \in [m-1]$. Let $A\in\calA(C,V)$ be an approval profile with $A(v)\in \{a_{\ell+1}^*, b_{\ell+1}^*, \sigma^*(a_{\ell+1}^*), \sigma^*(b_{\ell+1}^*)\}$ for all $v\in V$. Then:
\begin{align*}
\score{f}(W_1^*, A) > \score{f}(W_2^*, A) \implies W_1^* \succ_{\calF(A)} W_2^* \textrm{.}
\end{align*}
\end{lemma} 
\begin{proof}
We start by noting that if $b_{\ell+1}^*$ and $a_{\ell+1}^*$ are defined, then Case (C) occurred when defining $\gamma_{\ell+1}$. In particular, $t^*_{\ell+1}$ has been defined and Lemma~\ref{lemma:threshold} is applicable.

First we show that if $A$ contains both $a_{\ell+1}^*$ and $\sigma^*(a_{\ell+1}^*)$, then after removing both from $A$ the relative order of $W_1^*$ and $W_2^*$ does not change. Without loss of generality, let us assume that $W_1^* \succ_{\calF(A)} W_2^*$ and consider the profile $Q$ that consist of two votes, $a_{\ell+1}^*$ and $\sigma^*(a_{\ell+1}^*)$. By neutrality, $W_1^*$ and $W_2^*$ are equally good with respect to $Q$. If $W_2^* \succeq_{\calF(A - Q)} W_1^*$, then by consistency we would get that $W_2^* \succeq_{\calF(A)} W_1^*$, a contradiction. By the same argument we observe that if $A$ contains $b_{\ell+1}^*$ and $\sigma^*(b_{\ell+1}^*)$, then after removing them from $A$ the relative order of $W_1^*$ and $W_2^*$ does not change. 
Further if $A$ contains only votes $b_{\ell+1}^*$ and $a_{\ell+1}^*$, then by consistency we can infer that $W_1^*$ is preferred over $W_2^*$ in~$A$. Also, $A$ cannot contain only votes $\sigma^*(b_{\ell+1}^*)$ and $\sigma^*(a_{\ell+1}^*)$, since in both these single-vote profiles the score of $W_2^*$ is greater than the score of $W_1^*$ (this follows from Lemma~\ref{lemma:on_l_regular_profiles} and from the fact that $f$ for $\ell$-regular profiles is a linear transformation of an appropriate counting function $f_\ell$). 

The above reasoning shows that without loss of generality we can assume that in $A$ there are either only the votes of types $b_{\ell+1}^*$ and $\sigma^*(a_{\ell+1}^*)$ or only the votes of types $a_{\ell+1}^*$ and $\sigma^*(b_{\ell+1}^*)$. Let us consider the first case, and let us assume that in $A$ there are $y_A$ votes of type $b_{\ell+1}^*$ and $x_A$ votes of type $\sigma^*(a_{\ell+1}^*)$. Since $\score{f}(W_1^*, A) > \score{f}(W_2^*, A)$, we get that:
\begin{align*}
&y_A \cdot \score{f}(W_1^*, b_{\ell+1}^*) + x_A \cdot \score{f}(W_1^*, \sigma^*(a_{\ell+1}^*)) > 
y_A \cdot \score{f}(W_2^*, b_{\ell+1}^*) + x_A \cdot \score{f}(W_2^*, \sigma^*(a_{\ell+1}^*)) \text{.}
\end{align*}
Thus, from the definition of $\sigma^*$ we get that:
\begin{align*}
y_A \cdot \score{f}(W_1^*, b_{\ell+1}^*) + x_A \cdot \score{f}(W_2^*, a_{\ell+1}^*) >  y_A \cdot \score{f}(W_2^*, b_{\ell+1}^*) + x_A \cdot \score{f}(W_1^*, a_{\ell+1}^*) \text{.}
\end{align*}
Which is equivalent to:
\begin{align*}
x_A \cdot \big(\score{f}(W_1^*, a_{\ell+1}^*) - \score{f}(W_2^*, a_{\ell+1}^*) \big) < y_A \cdot \big( \score{f}(W_1^*, b_{\ell+1}^*) - \score{f}(W_2^*, b_{\ell+1}^*) \big) \text{.}
\end{align*}
From the above inequality we get that:
\begin{align*}
\frac{x_A}{y_A} < \frac{\score{f}(W_1^*, b_{\ell+1}^*) - \score{f}(W_2^*, b_{\ell+1}^*)}{\score{f}(W_1^*, a_{\ell+1}^*) - \score{f}(W_2^*, a_{\ell+1}^*)}
= \frac{\score{f}(W_1^*, b_{\ell+1}^*) - \score{f}(W_2^*, b_{\ell+1}^*)}{\gamma_{\ell+1}\big(\score{f_{\ell+1}}(W_1^*, a_{\ell+1}^*) - \score{f_{\ell+1}}(W_2^*, a_{\ell+1}^*)\big)} = t^*_{\ell+1} \textrm{.}
\end{align*}
Observe that $A = S(x_A, y_A)$, so since $\nicefrac{x_A}{y_A} < t^*_{\ell+1}$, from Lemma~\ref{lemma:threshold} we infer that $W_1^* \succ_{\calF(A)} W_2^*$.

Now, let us assume that $A$ consists only of the votes of types $a_{\ell+1}^*$ and $\sigma^*(b_{\ell+1}^*)$. 
In such case the profile $\sigma^*(A)$ consists only of votes of types $b_{\ell+1}^*$ and $\sigma^*(a_{\ell+1}^*)$. Further, $\score{f}(W_2^*, \sigma^*(A)) > \score{f}(W_1^*, \sigma^*(A))$. Similarly as before, let us assume that in $\sigma^*(A)$ there are $y_A$ votes of type $b_{\ell+1}^*$ and $x_A$ votes of type $\sigma^*(a_{\ell+1}^*)$. By similar reasoning as before we infer that $\nicefrac{x_A}{y_A} > t^*_{\ell+1}$, and by Lemma~\ref{lemma:threshold} that $W_2^* \succ_{\calF(\sigma^*(A))} W_1^*$. From this, by neutrality, it follows that $W_1^* \succ_{\calF(A)} W_2^*$, which completes the proof.
\end{proof}

Next, we generalize Lemma~\ref{lemma:specificVotes} to arbitrary profiles, yet we still focus on comparing the two distinguished profiles $W_1^*$ and $W_2^*$.

\begin{lemma}\label{lemma:WqW2Fixed}
For all $A\in \calA(C,V)$ it holds that
\begin{align*}
\score{f}(W_1^*, A) > \score{f}(W_2^*, A) \implies W_1^* \succ_{\calF(A)} W_2^* \textrm{.}
\end{align*}
\end{lemma} 
\begin{proof}
We prove this statement by induction on $\ell$-bounded profiles.
For $0$-bounded profiles $A$ this is trivial since $\score{f}(W_1^*, A) > \score{f}(W_2^*, A)$ cannot hold.

Assume that the statement holds for $\ell$-bounded profiles and assume that $\score{f}(W_1^*, A) > \score{f}(W_2^*, A)$.
If Case (A) was applicable when defining $\gamma_{\ell+1}$, i.e., if $\gamma_{\ell+1}=0$, then $\score{f}(W_1^*, A) > \score{f}(W_2^*, A)$ implies $\score{f}(W_1^*, \bounded(A,\ell)) > \score{f}(W_2^*, \bounded(A,\ell))$ since the score of $(\ell+1)$-regular profiles is $0$. 
This implies by the induction hypothesis that $W_1^* \succ_{\calF(\bounded(A,\ell))} W_2^*$.
Furthermore, since Case~(A) was applicable, $W_1^* \sim_{\calF(\regular(A,\ell+1))} W_2^*$.
Since $A=\bounded(A,\ell)+\regular(A,\ell+1)$, consistency yields that $W_1^* \succ_{\calF(A)} W_2^*$.

In Case (B), we know that $W_1^* \sim_{\calF(A)} W_2^*$ for all $\ell$-bounded profiles.
Hence $W_1^* \sim_{\calF(\bounded(A,\ell))} W_2^*$.
By our induction hypothesis, this implies that $\score{f}(W_1^*, \bounded(A,i)) = \score{f}(W_2^*, \bounded(A,i))$.
Hence $\score{f}(W_1^*, \regular(A,\ell+1)) > \score{f}(W_2^*, \regular(A,\ell+1))$.
Recall that Lemma~\ref{lemma:on_l_regular_profiles} states that $f_{\ell+1}$ implements $\calF$ on $(\ell+1)$-regular profiles.
Since $\regular(A,\ell+1)$ is an $(\ell+1)$-regular profile and $f(x,\ell+1)=f_{\ell+1}(x,\ell+1)$, in particular $\score{f}(W_1^*, \regular(A,\ell+1)) > \score{f}(W_2^*, \regular(A,\ell+1))$ implies $W_1^* \succ_{\calF(\regular(A,\ell+1))} W_2^*$.
Furthermore, by consistency, $W_1^*$ has the same relative position as $W_2^*$ in $\calF(\regular(A,\ell+1))$ and $\calF(A)$, which in turn implies $W_1^* \succ_{\calF(A)} W_2^*$.

In Case (C), for the sake of contradiction let us assume that $W_2^* \succeq_{\calF(A)} W_1^*$.
Let us take an arbitrary vote $v\in V$ with $A(v)\notin \{b_{\ell+1}^*, a_{\ell+1}^*, \sigma^*(b_{\ell+1}^*), \sigma^*(a_{\ell+1}^*)\}$.
We will show in the following that there exists a profile $A'$ with $\set(A')=\set(A)\setminus \{A(v)\}$, $\score{f}(W_1^*, A') > \score{f}(W_2^*, A')$, and $W_2^* \succeq_{\calF(A')} W_1^*$.
We then repeat this step until we obtain a profile $A''$ with $\set(A'')=\{b_{\ell+1}^*, a_{\ell+1}^*, \sigma^*(b_{\ell+1}^*), \sigma^*(a_{\ell+1}^*)\}$.
Still, it holds that $\score{f}(W_1^*, A'') > \score{f}(W_2^*, A'')$ and $W_2^* \succeq_{\calF(A'')} W_1^*$, but that contradicts Lemma~\ref{lemma:specificVotes}.
Consequently, $W_1^* \succ_{\calF(A)} W_2^*$ has to hold.

Let us now show that there exists a profile $A'$ with $\set(A')=\set(A)\setminus \{A(v)\}$, $\score{f}(W_1^*, A') > \score{f}(W_2^*, A')$, and $W_2^* \succeq_{\calF(A')} W_1^*$.
If $W_1^* \sim_{\calF(A(v))} W_2^*$, then by consistency the relative order of $W_1^*$ and $W_2^*$ in $\calF(A')$ is the same as in $\calF(A)$. Also, since the scores of committees $W_1^*$ and $W_2^*$ are the same in $v$ (cf. Lemma~\ref{lemma:on_l_regular_profiles}), we get that $\score{f}(W_1^*, A') > \score{f}(W_2^*, A')$.

Let us now consider the case that $W_1^*\succ_{\calF(A(v))} W_2^*$. Let $n_v = |\{v'\in V: A(v') = A(v)\}|$. We set
\begin{align}\label{eq:scoreW1W2Epsilon}
\epsilon = \score{f}(W_1^*, A) - \score{f}(W_2^*, A) > 0 \textrm{.}
\end{align}
We distinguish two cases: $|A(v)|\leq\ell$ and $|A(v)|=\ell+1$. Let us consider $|A(v)|\leq\ell$ first.
We observe that there exist values $x, y \in \naturals$ such that:
\begin{align}\label{eq:scoreW1W2}
\begin{split}
0 < \frac{x}{y} \big(\score{f}(W_1^*, \sigma^*(b_{\ell+1}^*)) - \score{f}(W_2^*, \sigma^*(b_{\ell+1}^*))\big) + 
n_v \big(\score{f}(W_1^*, v) - \score{f}(W_2^*, v)\big) < \frac{\epsilon}{2} \textrm{.}
\end{split}
\end{align}
Now, consider a profile $B = y\cdot A+x\cdot \sigma^*(b_{\ell+1}^*)+x\cdot b_{\ell+1}^*$. By consistency, $W_2^* \succeq_{\calF(B)} W_1^*$. Next, let us consider a profile $Q= x\cdot \sigma^*(b_{\ell+1}^*) + y\cdot n_v\cdot A(v)$. From Equality~\eqref{eq:scoreW1W2} we see that $W_1^*$ has a higher score in $Q$ than $W_2^*$.
Since $Q$ is $\ell$-bounded, by our inductive assumption we get that $W_1^* \succ_{\calF(Q)} W_2^*$. Consequently, by consistency we get that $W_2^*\succ_{\calF(B-Q)}W_1^*$ since otherwise $W_1^* \succ_{\calF(B)} W_2^*$, a contradiction. Further, from Equalities~\eqref{eq:scoreW1W2Epsilon} and \eqref{eq:scoreW1W2} we get that in $B-Q$ the score of $W_1^*$ is greater than the score of $W_2^*$, which can be seen as follows:
\begin{align*}
&\score{f}(W_1^*, B-Q) - \score{f}(W_2^*, B-Q) \\
& \hspace{2cm} = \score{f}(W_1^*, B) - \score{f}(W_2^*, B) - (\score{f}(W_1^*, Q) - \score{f}(W_2^*, Q)) \\
& \hspace{2cm} = y\epsilon - (\score{f}(W_1^*, Q) - \score{f}(W_2^*, Q)) > \frac{y \epsilon}{2} \textrm{.}
\end{align*}
We obtained the profile $B-Q = y\cdot A + x(\sigma^*(b_{\ell+1}^*) + b_{\ell+1}^*)-x\cdot \sigma^*(b_{\ell+1}^*) - y\cdot n_v\cdot A(v) = y\cdot (A-n_v \cdot A(v)) + x\cdot b_{\ell+1}^*$, for which $\set(B-Q)=\set(A)\setminus\{A(v)\}$.
Furthermore, the relative order of $W_1^*$ and $W_2^*$ in $\calF(B-Q)$ is the same as in $\calF(A)$, and $\score{f}(W_1^*, B-Q) > \score{f}(W_2^*, B-Q)$.

Let us now turn to the case that $|A(v)|=\ell+1$.
Similar to before, we choose $x, y \in \naturals$ such that:
\begin{align}\label{eq:scoreW1W2-2}
\begin{split}
0 < \frac{x}{y} \big(\score{f}(W_1^*, \sigma^*(a_{\ell+1}^*)) - \score{f}(W_2^*, \sigma^*(a_{\ell+1}^*))\big) + 
n_v \big(\score{f}(W_1^*, v) - \score{f}(W_2^*, v)\big) < \frac{\epsilon}{2} \textrm{.}
\end{split}
\end{align}
Now, consider a profile $B = y\cdot A+x\cdot \sigma^*(a_{\ell+1}^*)+x\cdot a_{\ell+1}^*$ for which, by consistency, $W_2^* \succeq_{\calF(B)} W_1^*$ holds. Let $Q= x\cdot \sigma^*(a_{\ell+1}^*) + y\cdot n_v\cdot A(v)$. From Equality~\eqref{eq:scoreW1W2-2} we see that $W_1^*$ has a higher score in $Q$ than $W_2^*$.
Since $Q$ is $(\ell+1)$-regular, Lemma~\ref{lemma:on_l_regular_profiles} gives us that $W_1^* \succ_{\calF(Q)} W_2^*$. As before, by consistency we get that $W_2^*\succ_{\calF(B-Q)}W_1^*$, and from Equalities~\eqref{eq:scoreW1W2Epsilon} and \eqref{eq:scoreW1W2-2} we get that 
$\score{f}(W_1^*, B-Q) > \score{f}(W_2^*, B-Q)$.
Hence, also in this case, we have obtained the profile $B-Q$, for which $\set(B-Q)=\set(A)\setminus\{A(v)\}$, the relative order of $W_1^*$ and $W_2^*$ in $\calF(B-Q)$ is the same as in $\calF(A)$, and $\score{f}(W_1^*, B-Q) > \score{f}(W_2^*, B-Q)$.

Finally, if $W_2^*\succ_{\calF(A(v))} W_1^*$ in $v$, we can repeat the above reasoning, but applying $\sigma*$ to all occurrences of  $b_{\ell+1}^*$, $a_{\ell+1}^*$, $\sigma^*(b_{\ell+1}^*)$, and $\sigma^*(a_{\ell+1}^*)$.
\end{proof}

Before we proceed further, we establish the existence of two particular profiles $A_\ell^*$ and $B_\ell^*$, that we will need for proving the most general variant of our statement.

\begin{lemma}\label{lemma:allContainingTheCandidate}
Let $W_1, W_2, W_3\in\pow_k(C)$ such that $|W_1 \cap W_3| > |W_1 \cap W_2|$. For each $\ell$, $1 \leq \ell \leq m$, if $\calF$ is non-trivial for $\ell$-regular profiles, then there exist two $\ell$-regular profiles, $A_\ell^*$ and $B_\ell^*$, such that:
\begin{enumerate}
\item $\score{f}(W_1,A_\ell^*)=\score{f}(W_3,A_\ell^*)>\score{f}(W_2,A_\ell^*)$ and $W_1\sim_{\calF(A_\ell^*)} W_3 \succ_{\calF(A_\ell^*)} W_2$,
\item $\score{f}(W_1,B_\ell^*)=\score{f}(W_3,B_\ell^*)<\score{f}(W_2,B_\ell^*)$ and $W_1\sim_{\calF(B_\ell^*)} W_3 \prec_{\calF(B_\ell^*)} W_2$.
\end{enumerate}
\end{lemma}
\begin{proof}
Let $c$ be a candidate such that $c \in W_1 \cap W_3$ and $c \notin W_2$. Such a candidate exists because $|W_1 \cap W_3| > |W_1 \cap W_2|$.
Profile $A_\ell^*$ contains, for each $S\subseteq C \setminus \{c\}$ with $|S|=\ell-1$, a vote with approval set $S \cup \{c\}$. 
First, let us note that all committees that contain $c$ have the same $f_{\ell}$-score in $A_\ell^*$: this follows from  neutrality, since the profile $A_\ell^*$ is symmetric with respect to committees containing $c$, in particular $W_1$ and $W_3$.
Let $s$ denote the score of such committees.

Next, we will argue that $\score{f_{\ell}}(W_2,A_\ell^*)<s$. To see this, let $c' \in W_2$ and consider a committee $W_2' = (W_2 \setminus \{c'\}) \cup \{c\}$.
Since $f$ implements $\calF$, there exists $x \leq k$ such that $f_\ell(x, \ell) > f_\ell(x-1, \ell)$.
Due to Proposition~\ref{prop:counting-functions-equivalent} we can assume that $m-\ell\geq k-(x-1)$; otherwise this difference between $f_\ell(x, \ell)$ and $f_\ell(x-1, \ell)$ would not be relevant for computing scores.
Let $T\subseteq C\setminus \{c,c'\}$ such that $|T|=\ell-1$ and $|T\cap W_2|=x-1$.
To show that such a $T$ exists, we have to prove that there exist $(\ell-1)-(x-1)$ candidates in ($C\setminus W_2)\setminus \{c,c'\}$. This is the case since $m-\ell\geq k-(x-1)$ and thus $|(C\setminus W_2)\setminus \{c,c'\}|=m-k-1\geq \ell-x$.

Now let $v$ be the vote in $A_\ell^*$ with approval set $T\cup \{c\}$.
Since $f_\ell(x, \ell) > f_\ell(x-1, \ell)$, \[f_\ell(|A_\ell^*(v)\cap W_2'|,|A_\ell^*(v)|)> f_\ell(|A_\ell^*(v)\cap W_2|,|A_\ell^*(v)|).\]
Furthermore, for all votes $v'$ in $A_\ell^*$:
\[f_\ell(|A_\ell^*(v')\cap W_2'|,|A_\ell^*(v)|)\geq f_\ell(|A_\ell^*(v')\cap W_2|,|A_\ell^*(v)|).\]
Hence, $\score{f_{\ell}}(W_2',A_\ell^*)> \score{f_{\ell}}(W_2,A_\ell^*)$.
Since $f(x,\ell)=\gamma_\ell\cdot f_\ell(x,\ell)$ we get $\score{f}(W_2',A_\ell^*)> \score{f}(W_2,A_\ell^*)$.
Further, by a previous argument we have $\score{f}(W_1,A_\ell^*)=\score{f}(W_2',A_\ell^*)$, thus by transitivity 
we conclude that $\score{f}(W_1,A_\ell^*)>\score{f}(W_2,A_\ell^*)$.

Next, let us construct profile $B_\ell^*$. In this case we choose $c$ such that $c \in W_2$ and $c \notin W_1 \cup W_3$. Again, this is possible because $|W_3\setminus W_1| = k - |W_1 \cap W_3| < k - |W_1 \cap W_2| = |W_2\setminus W_1|$ and hence $W_2 \not\subseteq W_1\cup W_3$.
Similarly as before, $B_\ell^*$ contains a vote with approval set $S \cup \{c\}$ for each $S\subseteq C \setminus \{c\}$ with $|S|=\ell-1$. With similar arguments as before we can show that all committees that contain $c$ have the same score in $B_\ell^*$ (in particular $W_2$) and this score is larger than the score of committees that do not contain $c$ (in particular $W_1$ and $W_3$).

Finally, the statements concerning $\calF$ follow from Lemma~\ref{lemma:on_l_regular_profiles} since both $A_\ell^*$ and $B_\ell^*$ are $\ell$-regular.
\end{proof}

We further generalize Lemma~\ref{lemma:specificVotes} and \ref{lemma:WqW2Fixed} so to allow us to compare $W_1^*$ with arbitrary profiles. This is the final step; we can then proceed with a direct proof of Theorem~\ref{thm:characterizationWelfareFunctions}.

\begin{lemma}\label{lemma:arbitraryCommittees}
For all $A\in \calA(C,V)$ and $W\in \pow_k(C)$ it holds that
\begin{align*}
\score{f}(W_1^*, A) > \score{f}(W, A) \implies W_1^* \succ_{\calF(A)} W \textrm{.}
\end{align*}
\end{lemma}
\begin{proof}
We prove this statement by induction on $\ell$-bounded profiles.
As in Lemma~\ref{lemma:WqW2Fixed}, for $0$-bounded profiles $A$ the statement is trivial since $\score{f}(W_1^*, A) > \score{f}(W, A)$ cannot hold.

In order to prove the inductive step, we assume that the statement holds for $\ell$-bounded profiles. Let $A$ be an $(\ell + 1)$-bounded profile and assume that $\score{f}(W_1^*, A) > \score{f}(W, A)$. We will show that $W_1^* \succ_{\calF(A)} W$.
If Case (A) or (B) was applicable when defining $\gamma_{\ell+1}$, the same arguments as in Lemma~\ref{lemma:WqW2Fixed} yield that $W_1^* \succ_{\calF(A)} W$.

If Case (C) was applicable when defining $\gamma_{\ell+1}$ and if $|W_1^* \cap W| = |W_1^* \cap W_2^*|$, then the statement of the lemma follows from Lemma~\ref{lemma:WqW2Fixed} and neutrality.
Recall that we fixed $W_1^*$ and $W_2^*$ as two committees with the smallest possible size of the intersection.
Thus, if $|W_1^* \cap W| \neq |W_1^* \cap W_2^*|$ then $|W_1^* \cap W| > |W_1^* \cap W_2^*|$.
For the sake of contradiction let us assume that $W \succeq_A W_1^*$. 
Let $\score{f}(W_1^*, A) - \score{f}(W, A) = \epsilon > 0$.

Now, from $A$ we create a new profile $B$ in the following way.
Let us consider two cases:
\begin{description}
\item[Case 1:] $\score{f}(W_2^*, \most(A, \ell)) - \score{f}(W, \most(A, \ell)) \geq 0$.

Let $Q$ be an $\ell$-bounded profile where:
\begin{align*}
\score{f}(W_1^*, Q) = \score{f}(W, Q) > \score{f}(W_2^*, Q) \textrm{.}
\end{align*}
Such a profile exists due to Lemma~\ref{lemma:allContainingTheCandidate}. Since $\score{f}(W_2^*, Q) - \score{f}(W, Q)$ is negative, there exist such $x\in \naturals$, $y \in \naturals\cup\{0\}$ that $x \geq 2$ and
\begin{align*}
&0 \leq \Big(\score{f}(W_2^*, \most(A, \ell)) - \score{f}(W, \most(A, \ell)) \Big) \\
& \hspace{2cm } + \nicefrac{y}{x} \cdot \Big(\score{f}(W_2^*, Q) - \score{f}(W, Q) \Big) < \nicefrac{\epsilon}{2} \textrm{,}
\end{align*}
which is equivalent to
\begin{align}\label{eq:ww2relation}
0 \leq  \score{f}(W_2^*, x \most(A, \ell) + yQ) - \score{f}(W, x \most(A, \ell) + yQ) < \nicefrac{x\epsilon}{2} \textrm{.}
\end{align}
We set $B = xA + yQ$.
\item[Case 2:] $\score{f}(W_2^*, \most(A, \ell)) - \score{f}(W, \most(A, \ell)) < 0$.

In this case our reasoning is very similar. Let $Q$ be an $\ell$-bounded profile where:
\begin{align*}
\score{f}(W_2^*, Q) > \score{f}(W_1^*, Q) = \score{f}(W, Q) \textrm{.}
\end{align*}
Again, similarly as before, we observe that there exist such $x, y \in \naturals$ that $x \geq 1$ and:
\begin{align*}
&0 \leq \Big(\score{f}(W_2^*, \most(A, \ell)) - \score{f}(W, \most(A, \ell)) \Big) \\
& \hspace{2cm } + \nicefrac{y}{x} \cdot \Big(\score{f}(W_2^*, Q) - \score{f}(W, Q) \Big) < \nicefrac{\epsilon}{2} \textrm{,}
\end{align*}
which is equivalent to Inequality~\eqref{eq:ww2relation}. Here, we also set $B = xA + yQ$.
\end{description}

By similar transformation as before, but applied to $\exact(B, \ell+1)$ rather than to $\most(B, \ell)$, we construct a profile $D$ from $B$:
\begin{description}
\item[Case 1:] $\score{f}(W_2^*, \exact(B, \ell+1)) - \score{f}(W, \exact(B, \ell+1)) \geq 0$.

Due to Lemma~\ref{lemma:allContainingTheCandidate} there exists an $(\ell+1)$-regular profile $Q'$ with
\begin{align*}
\score{f}(W_1^*, Q') = \score{f}(W, Q') > \score{f}(W_2^*, Q') \textrm{.}
\end{align*}
Similarly as before, there exist $x'\in \naturals$, $y' \in \naturals\cup\{0\}$ such that
\begin{align}\label{eq:ww2relation2}
0 \leq  \score{f}(W_2^*, x' \exact(B, \ell+1) + y'Q') - \score{f}(W, x' \exact(B, \ell+1) + y'Q') < \nicefrac{x'\epsilon}{2} \textrm{.}
\end{align}
We set $D = x'B + y'Q'$.
\item[Case 2:] $\score{f}(W_2^*, \exact(A, \ell+1)) - \score{f}(W, \exact(A, \ell+1)) < 0$.

Here, let $Q'$ be an $(\ell+1)$-regular profile such that
\begin{align*}
\score{f}(W_1^*, Q') = \score{f}(W, Q') > \score{f}(W_2^*, Q') \textrm{.}
\end{align*}
There exist $x', y' \in \naturals$ such that Inequality~\eqref{eq:ww2relation2} is satisfied. We set $D = x'B + y'Q'$.
\end{description}

Let us analyze the resulting profile $D = x'xA + x'yQ + y'Q'$. By our assumption we know that $W \succeq_A W_1^*$, thus by consistency we get that $W \succeq_{xx'A} W_1^*$. Since $W \sim_{\calF(Q)} W_1^*$ and $W \sim_{\calF(Q')} W_1^*$ due to Lemma~\ref{lemma:allContainingTheCandidate}, from consistency it follows that $W \succeq_{\calF(D)} W_1^*$.

Further, since $Q$ is $\ell$-bounded and $Q'$ is $(\ell+1)$-regular, 
\begin{align*}
D &= x'xA + x'yQ + y'Q' \\
  &= \most(x'xA + x'yQ + y'Q', \ell) + \exact(x'B + y'Q', \ell+1) \\
  &= \most(x'xA + x'yQ, \ell) + \exact(x'B + y'Q', \ell+1) \\
  &= x'\most(xA + yQ, \ell) + \exact(x'B + y'Q', \ell+1)\textrm{.}
\end{align*}
Inequalities~\eqref{eq:ww2relation}~and~\eqref{eq:ww2relation2} imply that $W_2^*$ has a higher score than $W$ in profiles $x'(x\most(A, \ell) + yQ) = x'\most(xA + yQ, \ell)$ and $x'\exact(B, \ell+1)+ y'Q' = \exact(x'B + y'Q', \ell+1)$. From our inductive assumption we get that $W_2^*$ is preferred over $W$ in $x'\most(xA + yQ, \ell)$, and by Lemma~\ref{lem:implements-on-approval-profiles} we get that $W_2^*$ is preferred over $W$ in $\exact(x'B + y'Q', \ell+1)$. Consistency implies that $W_2^* \succeq_{\calF(D)} W$, and thus $W_2^*\succeq_{\calF(D)} W \succeq_{\calF(D)} W_1^*$.

Now we observe that
\begin{align*}
&\score{f}(W_1^*, \most(xA + yQ, \ell)) - \score{f}(W_2^*, \most(xA + yQ, \ell)) \\
& \hspace{1cm} = \Big(\score{f}(W_1^*, \most(xA + yQ, \ell)) - \score{f}(W, \most(xA + yQ, \ell))\Big) \\
& \hspace{2cm} + \Big(\score{f}(W, \most(xA + yQ, \ell)) - \score{f}(W_2^*, \most(xA + yQ, \ell))\Big) \\
& \hspace{1cm} \geq \Big(\score{f}(W_1^*, \most(xA + yQ, \ell)) - \score{f}(W, \most(xA + yQ, \ell))\Big) - \frac{x\epsilon}{2} \\
& \hspace{1cm} = \Big(\score{f}(W_1^*, \most(xA, \ell)) - \score{f}(W, \most(xA, \ell))\Big) - \frac{x\epsilon}{2} \text{.}
\end{align*}
and
\begin{align*}
&\score{f}(W_1^*, \exact(x'B + y'Q', \ell+1)) - \score{f}(W_2^*, \exact(x'B + y'Q', \ell+1)) \\
& \hspace{1cm} = \Big(\score{f}(W_1^*, \exact(x'B + y'Q', \ell+1)) - \score{f}(W, \exact(x'B + y'Q', \ell+1))\Big) \\
& \hspace{2cm} + \Big(\score{f}(W, \exact(x'B + y'Q', \ell+1)) - \score{f}(W_2^*, \exact(x'B + y'Q', \ell+1))\Big) \\
& \hspace{1cm} \geq \Big(\score{f}(W_1^*, \exact(x'B + y'Q', \ell+1)) - \score{f}(W, \exact(x'B + y'Q', \ell+1))\Big) - \frac{x'\epsilon}{2} \\
& \hspace{1cm} = \Big(\score{f}(W_1^*, \exact(x'B, \ell+1)) - \score{f}(W, \exact(x'B, \ell+1))\Big) - \frac{x'\epsilon}{2} \\
& \hspace{1cm} = \Big(\score{f}(W_1^*, \exact(x'xA, \ell+1)) - \score{f}(W, \exact(x'xA, \ell+1))\Big) - \frac{x'\epsilon}{2} \text{.}
\end{align*}
By combining the above two inequalities we get that
\begin{align*}
&\score{f}(W_1^*, D) - \score{f}(W_2^*, D) \\
& \hspace{1cm} = x'\cdot\Big(\score{f}(W_1^*, \most(xA + yQ, \ell)) - \score{f}(W, \most(xA + yQ, \ell))\Big) \\
& \hspace{2cm} +  \Big(\score{f}(W_1^*, \exact(x'B + y'Q', \ell+1)) - \score{f}(W, \exact(x'B + y'Q', \ell+1))\Big) \\
& \hspace{1cm} \geq x'\cdot \Big(\score{f}(W_1^*, \most(xA, \ell)) - \score{f}(W, \most(xA, \ell))\Big) \\
& \hspace{2cm} +  \Big(\score{f}(W_1^*, \exact(x'xA, \ell+1)) - \score{f}(W, \exact(x'xA, \ell+1))\Big) - \frac{(x'+xx')\epsilon}{2} \\
& \hspace{1cm} = xx'\cdot \Big(\score{f}(W_1^*, A) - \score{f}(W, A)\Big) - \frac{(x'+xx')\epsilon}{2} \\
& \hspace{1cm} = xx'\epsilon - \frac{(x'+xx')\epsilon}{2} = \frac{(xx' - x')\epsilon}{2} > 0 \text{.}
\end{align*}

\noindent
Summarizing, we obtained a profile $D$, such that
$\score{f}(W_1^*, D) > \score{f}(W_2^*, D)$ and
$W_2^* \succ_{\calF(D)} W_1^*$.
This, however, contradicts Lemma~\ref{lemma:WqW2Fixed}. Hence, we have proven the inductive step, which completes the proof of the lemma.
\end{proof}

Lemma~\ref{lemma:arbitraryCommittees} allows us to prove Theorem~\ref{thm:characterizationWelfareFunctions}, our characterization of ABC counting rules.

\begin{proof}[Finalizing the proof of Theorem~\ref{thm:characterizationWelfareFunctions}]
Let $\calF$ satisfy symmetry, consistency, weak efficiency, and continuity.
If $\calF$ is trivial, then $f(x,y)=0$ implements $\calF$.

If $\calF$ is non-trivial, we construct $f$, $W_1^*$, and $W_2^*$ as described above.
We claim that for $A\in\calA(C,V)$ and $W_1, W_2\in \pow_k(C)$ it holds that $\score{f}(W_1, A) > \score{f}(W_2, A)$ if and only if $W_1 \succ_{\calF(A)} W_2$.
By neutrality, Lemma~\ref{lemma:arbitraryCommittees} is applicable to any pair of committees $W_1, W_2\in \pow_k(C)$: if $\score{f}(W_1, A) > \score{f}(W_2, A)$ then $W_1 \succ_{\calF(A)} W_2$. 

Now, for the other direction, instead of showing that $W_1 \succ_{\calF(A)} W_2$ implies $\score{f}(W_1, A) > \score{f}(W_2, A)$, we show that $\score{f}(W_1, A) = \score{f}(W_2, A)$ implies $W_1 \sim_{\calF(A)} W_2$. Note that Lemma~\ref{lemma:arbitraryCommittees} does not apply to committees with the same score.
For the sake of contradiction let $\score{f}(W_1, A) = \score{f}(W_2, A)$ but $W_1 \succ_{\calF(A)} W_2$. 
As a first step, we prove that there exists a profile $B$ with $\score{f}(W_2, B) > \score{f}(W_1, B)$ and $W_2 \succ_{\calF(B)} W_1$.
Since $W_1 \succ_{\calF(A)} W_2$ and by neutrality, there exists a profile $A'\in\calA(C,V)$ with $W_2 \succ_{\calF(A')} W_1$.
Thus, there exists an $\ell\in[m]$ such that $W_2 \succ_{\calF(\regular(A',\ell))} W_1$, because otherwise, by consistency, $W_1 \succeq_{\calF(A')} W_2$ would hold; let $B=\regular(A',\ell)$.
Now, Lemma~\ref{lemma:on_l_regular_profiles} guarantees that $\score{f_\ell}(W_2,B)>\score{f_\ell}(W_1,B)$.
Since $f(x,\ell)=\gamma_\ell\cdot f_\ell(x,\ell)$, also $\score{f}(W_2,B)>\score{f}(W_1,B)$.
Observe that for each $n\in\naturals$ we have $\score{f}(W_2, B+nA) > \score{f}(W_1, B+nA)$. Thus, by Lemma~\ref{lemma:arbitraryCommittees} for each $n$, $W_2 \succ_{\calF(B+nA)} W_1$, which contradicts continuity of $\calF$.
Hence $\score{f}(W_1, A) = \score{f}(W_2, A)$ implies $W_1 \sim_{\calF(A)} W_2$ and, consequently, $\score{f}(W_1, A) > \score{f}(W_2, A)$ if and only if $W_1 \succ_{\calF(A)} W_2$.
We see that $f$ implements $\calF$ and thus $\calF$ is an ABC counting rule.

Finally, as we already noted, an ABC counting rule satisfies symmetry, consistency, weak efficiency, and continuity: this follows immediately from the definitions.
\end{proof}

\subsection{Independence of Axioms}
\label{subsec:independence}

The set of axioms used in the statement of Theorem~\ref{thm:characterizationWelfareFunctions} is minimal. %
First, let us consider the variation of AV where the score of a fixed candidate $c$ is doubled.
Formally, the score of a committee $W$ is defined as $\sum_{v \in V}|A(v) \cap W| + |\{v\in V: c\in A(v) \cap W\}|$.
This rule satisfies all axioms except for neutrality.
If we consider a variation of AV where voter $1$ has a weight of $2$, i.e., voter~$1$ gives a score of $2$ to each approved candidate; all other voters have a weight of $1$.
This weighted AV rule clearly fails anonymity, but satisfies all other axioms.
Note that here we need the fact that consistency only has to hold for disjoint voter sets.

Next, consider Proportional Approval Voting where ties are broken by Multi-Winner Approval Voting.
This rule---let us call it $\calF^*$---satisfies all axiom except for continuity: 
consider the profile $A=(\{c\})$ and $A'=(\{a,b\},\{a,b\},\{c\})$. It holds that $\{a,b\}\succ_{\calF^*(A')}\{a,c\}$ because the PAV-score of both committees is $3$, but the AV-score of $\{a,b\}$ is $4$ and only $3$ for $\{a,c\}$.
However, it holds that $\{a,c\}\succ_{\calF^*(A+nA')}\{a,b\}$ for arbitrary $n$ because the PAV-scores of $\{a,c\}$ and $\{a,b\}$ are $3n+1$ and $3n$, respectively.

To see that consistency is independent, consider an ABC ranking rule that is PAV on party-list profiles (i.e., D'Hondt) and the trivial rule otherwise.
This rule fails consistency, since the addition of two party-list profiles may not be a party-list profile. All other axioms are satisfied by it: symmetry and weak efficiency are easy to see, continuity follows from the fact that in non-party-list profiles all committees are winning.
Finally, the rule which reverses the output of Multi-Winner Approval Voting (i.e., $f(x,y)=-x$) satisfies all axioms except for weak efficiency.

\section{Further Proof Details}
\label{sec:app:proof-details}

\begin{repproposition}{prop:counting-functions-equivalent}
\propcountingfunctionsequivalent
\end{repproposition}

\begin{proof}
Let $A\in \calA(C,V)$ and $W\in\pow_k(C)$. Let $D\subseteq [0,k]\times[0,m]$ be the domain of $f$ and $g$ that is actually used in the computation of $\score{f}(W,A)$ and $\score{g}(W,A)$.
We will show that 
\begin{align}
D\subseteq D_{m,k}\cup \{(k,m)\}.\label{eq:domain-condition}
\end{align}
Let $v\in V$, $x=|A(v)\cap W|$, and $y=|A(v)|$.
If $y=m$, then $x=|A(v)\cap W|=k$ and condition~\eqref{eq:domain-condition} is satisfied.
Let $y<m$.
If $y$ is sufficiently large (close to $m$), then $A(v)\cap W$ cannot be empty.
More precisely, it has to hold that the number of not approved members of $W$, $k-x$, is at most equal to the total number of not approved candidates in $v$, $m-y$; this yields that $k-x \leq m-y$.
Furthermore, $x\leq y$ (the number of approved members of $W$ must be at most equal to the total number of approved candidates).
Consequently, $(x,y)\in D_{m,k}$.
This shows that condition~\eqref{eq:domain-condition} holds.

Consider functions $f$ and $g$ as in the statement of the proposition. We will now show that for all $W_1,W_2\in\pow_k(C)$, it holds that: \[\score{g}(W_1,A)-\score{g}(W_2,A)=c\cdot (\score{f}(W_1,A)-\score{f}(W_2,A)).\]
Let $V_i=\{v\in V:|A(v)|=i\}$ for $i\in[m]$.
Now
\begin{align*}
&\score{g}(W_1,A)-\score{g}(W_2,A) = \\
&= \sum_{i=1}^m \sum_{v \in V_i} g(|A(v) \cap W_1|, |A(v)|) - g(|A(v) \cap W_2|, |A(v)|)\\
&= \sum_{i=1}^{m-1} \sum_{v \in V_i} \Big( c\cdot f(|A(v) \cap W_1|, |A(v)|) + d(y) - c\cdot f(|A(v) \cap W_2|, |A(v)|) - d(y) \Big)\\
&= c\cdot \sum_{v \in V} \Big( f(|A(v) \cap W_1|, |A(v)|) - f(|A(v) \cap W_2|, |A(v)|) \Big)\\
&= c\cdot \left(\score{f}(W_1,A)-\score{f}(W_2,A)\right)
\end{align*}
Consequently, $\score{g}(W_1,A) > \score{g}(W_2,A)$ if and only if $\score{f}(W_1,A) > \score{f}(W_2,A)$, and so
$W_1\succ_{f(A)} W_2$ if and only if $W_1\succ_{g(A)} W_2$.
\end{proof}

\begin{reptheorem}{thm:pav-ranking-characterization}
\theorempavcharacterization
\end{reptheorem}
\begin{proof}
Theorem~\ref{thm:pav-ranking-characterization} is a special case of Theorem~\ref{thm:thiele-abc-rule-characterization}.
\end{proof}

\begin{reptheorem}{thm:approval-characterizationB}
\thmapprovalB
\end{reptheorem}
\begin{proof}
It is straightforward to verify that Multi-Winner Approval Voting satisfies disjoint equality.
For the other direction, consider an ABC counting rule satisfying disjoint equality that is implemented 
by a counting function $f$.
As in previous proofs we rely on Proposition~\ref{prop:counting-functions-equivalent} to show that $f$ and $f_\text{AV}(x,y)=x$ implement the same ABC counting rule. It is thus our aim to show that for $(x,y)\in D_{m,k}$ it holds that $f(x,y)=c\cdot x + d(y)$ for some $c\in\reals$ and $d\colon [m]\to\reals$.
More specifically, we will show that for $(x,y)\in D_{m,k}$ with $0\leq x< y$ it holds that $f(x+1,y)-f(x,y)= f(1,1)-f(0,1)$. It then follows from induction that $f(x,y)= (f(1,0)-f(0,0))\cdot x + f(0,y)$ and thus we will be able to conclude that $f$ implements Multi-Winner Approval Voting.

Let $(x,y)\in D_{m,k}$ with $x<k$ and $x< y$.
We construct a profile $A\in \calA(C,[k-x+1])$ with $|A(1)|=y$ and $|A(2)|=\dots=|A(k-x+1)|=1$. All voters have disjoint sets of approved candidates. Hence this construction requires $y+k-x$ candidates. Since $(x,y)\in D_{m,k}$, it holds that $k-x\leq m-y$ and hence $y+k-x\leq m$; we see that a sufficient number of candidates is available.
Let $W_1$ contain $x$ candidates from $A(1)$ and one candidate from $A(2),\dots, A(k-x+1)$ each.
Let $W_2$ contain $x+1$ candidates from $A(1)$ and one candidate from $A(2),\dots, A(k-x)$ each.
Note that $|W_1|=|W_2|=k$.
By disjoint equality both $W_1$ and $W_2$ are winning committees.
Hence \[f(x,y)+(k-x)\cdot f(1,1) = f(x+1,y)+(k-x-1)\cdot f(1,1)+f(0,1)\] and thus 
$f(x+1,y)-f(x,y) = f(1,1)-f(0,1)$. 
\end{proof}

\begin{reptheorem}{thm:cc_characterization}
\thmccabcrulecharacterization
\end{reptheorem}

\begin{proof}
The Approval Chamberlin--Courant rule maximizes the number of voters that have at least one approved candidate in the committee. In a party-list profile, this implies that the $k$ largest parties receive at least one representative in the committee and hence disjoint diversity is satisfied.

For the other direction, let $\calF$ be an ABC counting rule implemented by a counting function~$f$.
Recall Proposition~\ref{prop:counting-functions-equivalent} and the relevant domain of counting functions $D_{m,k}=\{(x,y) \in [0,k]\times[0,m-1]: x\leq y \mathbin{\wedge} k-x\leq m-y\}$.
In a first step, we want to show that $f(x+1, y) = f(x, y)$ for $x\geq 1$ and $(x+1, y),(x,y)\in D_{m,k}$.
Let us fix $(x,y)$ such that $(x, y)\in D_{m,k}$, $(x+1, y)\in D_{m,k}$, and $x \geq 1$. Furthermore, let us fix a committee $W$ and consider a set $X\subseteq C$ with $|X|=y$ and $|X\cap W|=x$.
We construct a party-list profile $A$ as follows: $A$ contains $\zeta$ votes that approve $X$ (intuitively, $\zeta$ is a large natural number); further for each candidate $c \in W\setminus X$, profile $A$ contains a single voter who approves $\{c\}$.
This construction requires $y+(k-x)$ candidates. Since $(x, y)\in D_{m,k}$, we have $y+(k-x)\leq m$.

If we apply disjoint diversity to profile $A$, we obtain a winning committee $W'$ with $W\setminus X\subseteq W'$ and $|W'\cap X|\geq 1$.
Observe that $\score{f}(W',A)=\score{f}(W,A)$ (the satisfaction of all voters remains the same).
Let $W''$ be the committee we obtain from $W$ by replacing one candidate in $W\setminus X$ with a candidate in $X\setminus W$ (such a candidate exists since $(x+1, y)\in D_{m,k}$).
Since $W$ is a winning committee, $\score{f}(W'',A)\leq \score{f}(W,A)$ and thus
\begin{align}
\zeta f(x+1, y) + (k-x -1) f(1, 1) \leq \zeta f(x, y) + (k-x) f(1, 1) \text{.}
\end{align}
The above condition can be written as $f(x+1, y) - f(x, y) \leq \frac{1}{\zeta} \cdot f(1, 1)$. Since this must hold for any $\zeta$, we get that $f(x+1, y) \leq f(x, y)$. Since $f(x+1, y) \geq f(x, y)$ by the definition of counting functions, we get that $f(x+1, y) = f(x, y)$ for $x\geq 1$. By Proposition~\ref{prop:counting-functions-equivalent} we can set $f(0, y) = 0$ for each $y \in [m]$.
We conclude that $\calF$ is also implemented by the counting function 
\begin{align*}
f_\alpha(x,y) = \begin{cases}0 & \text{if }x = 0,\\\alpha(y) & \text{if }x\geq 1.\end{cases}
\end{align*}

As a next step we show that for the counting function $f_\alpha(x,y)$ we can additionally assume that $\alpha(y) = \alpha(1)$, for each $y$. Observe that if $y \geq m-k+1$, then for each committee $W$, a voter who approves $y$ candidates in total, approves at least one member of $W$. By our previous reasoning, each committee gets from such a voter the same score, and so such a voter does not influence the outcome of an election. Consequently, 
we can assume that $\alpha(y)=\alpha(1)$ for $y>m-k$.
Now, for $y\leq m-k$, we also show that $\alpha(y)=\alpha(1)$.
Towards a contradiction assume that $\alpha(y)\neq \alpha(1)$ and further, without loss of generality, $\alpha(y)> \alpha(1)$. To this end, let $n$ be natural number large enough so that $(n-1)\cdot \alpha(y)> n\cdot \alpha(1)$. 
Consider a party-list profile consisting of $n-1$ voters approving $\{c_1,\dots,c_y\}$, and, for $j\in[k]$, $n$ voters each that approves candidate $\{c_{y+j}\}$.
The committee $W_1=\{c_{y+1},\dots,c_{y+k}\}$ obtains a score of $nk\cdot f(1,1)=nk\cdot \alpha(1)$, whereas $W_2=\{c_1, c_{y+2},\dots,c_{y+k}\}$ obtains a score of $(n-1)\cdot \alpha(y)+n(k-1)\cdot \alpha(1)$. Since by choice of $n$ it holds that $(n-1)\cdot \alpha(y)> n\cdot \alpha(1)$, committee $W_2$ is winning. This contradicts disjoint diversity and hence $\alpha(y)=\alpha(1)$.

Finally, we use Proposition~\ref{prop:counting-functions-equivalent} to argue that the CC counting function $f_\text{CC}$ implements $\calF$. 
We distinguish two cases: $\alpha(1)>0$ and $\alpha(1)=0$.
If  $\alpha(1)>0$, then
$f_\text{CC} = \frac{1}{\alpha(1)}\cdot f_\alpha(x,y)$,
and we see that Proposition~\ref{prop:counting-functions-equivalent} indeed applies.

If $\alpha(1)=0$, then $f_\alpha$ is equivalent (by Proposition~\ref{prop:counting-functions-equivalent}) to the trivial counting function $f_0(x,y)=0$. Since $\calF$ is non-trivial, this case cannot occur.
\end{proof}

\begin{reptheorem}{thm:thiele-abc-rule-characterization}
\thmthieleabcrulecharacterization
\end{reptheorem}
\begin{proof}
To see that the $w$-Thiele method satisfies $\mathbf{d}$-proportionality, 
let $f_{w\text{-T}}$ be the $w$-Thiele method's counting function defined by $f_{w\text{-T}}(x,y)=\sum_{i=1}^x w_i$.
Consider a party-list profile $A$ with $p$ parties, i.e., we have a partition of voters $N_1, N_2, \ldots N_p$ and their corresponding joint approval sets $C_1,\dots, C_p$. For the sake of contradiction let us assume that $W\in \pow_k(C)$ is a winning committee and that there exists $i, j$ such that $\frac{|N_i|}{d_{|W \cap C_i|}} < \frac{|N_j|}{d_{|W \cap C_j| + 1}}$, $W \cap C_i \neq \emptyset$ and $C_j \setminus W \neq \emptyset$. 
Let $a\in W \cap C_i$ and $b\in C_j \setminus W$.
We define $W'=W\cup \{b\}\setminus \{a\}$.
Let us compute the difference between PAV-scores of $W$ and $W'$:
\begin{align*}
\score{f_{w\text{-T}}}(W',A) - \score{f_{w\text{-T}}}(W,A) = \frac{-|N_i|}{d_{|W \cap C_i|}} + \frac{|N_j|}{d_{|W \cap C_j| + 1}} >  0 \textrm{.}
\end{align*}
Thus, we see that $W'$ has a higher $w$-score than $W$, a contradiction.

To show the other direction, let $\calF$ be an ABC counting rule that satisfies $\mathbf{d}$-proportionality and $f$ its corresponding counting function.
We intend to apply Proposition~\ref{prop:counting-functions-equivalent} to show that $f$ is equivalent to the $w$-Thiele method's counting function $f_{w\text{-T}}(x,y)=\sum_{i=1}^x w_i$.
Hence we have to show that there exists a constant $c$ and a function $d\colon[m]\to \reals$ such that $f(x)=c\cdot f_{w\text{-T}}(x,y) + d(y)$ for all $(x,y)\in D_{m,k}= \{(x,y) \in [0,k]\times[0,m-1]: x\leq y \wedge k-x\leq m-y\}$. 

Let us fix $x \in [k]$ such that $k-x < m-y$. We consider two cases: we start with the case when $d_x \neq \infty$.

\begin{enumerate}
\item[$d_x$ is a positive integer:] Let us consider the following party-list profile. There are $k-x + 2$ groups of voters: $N_1, \ldots, N_{k-x+2}$ with $|N_1| = d_x$, $|N_i| = d_1$ for $i \geq 2$; their corresponding approval sets are $C_1, \ldots, C_{k-x+2}$. Let $|C_1|=y$, $|C_i|=1$ for $i \in [2, k-x+1]$, and $|C_{k-x+2}|=m - y - k +x \geq 1$.
Consider the two following committees: 
we choose $W_1$ such that $|W_1\cap C_1|=x-1$, $|W_1\cap C_i|=1$ for $i \geq 2$;
we chose $W_2$ such that $|W_2\cap C_1|=x$, $|W_2\cap C_2|=0$, and $|W_2\cap C_i|=1$ for $i \geq 1$. 

It is straight-forward to verify that both $W_1$ and $W_2$ are $\mathbf{d}$-proportional.
\iffalse
Let us start by showing that $W_1$ is $\mathbf{d}$-proportional for $x\geq 2$. Let $i \geq 2$:
\begin{align}
&\frac{|N_1|}{d_{|W_1 \cap C_1|}} = \frac{d_x}{d_{x-1}} \geq 1 \geq \frac{d_1}{d_2}  = \frac{|N_i|}{d_{|W_1 \cap C_i| + 1}} \textrm{,} \\
&\frac{|N_i|}{d_{|W_1 \cap C_i|}} = \frac{d_1}{d_{1}} = 1 = \frac{d_x}{d_x} = \frac{|N_1|}{d_{|W_1 \cap C_1| + 1}} \textrm{.}
\end{align}
If $x = 1$ we can omit the first from the above inequalities.

For $W_2$ and $x\geq 1$ we have:
\begin{align*}
&W_2 \cap C_2 = \emptyset \textrm{,} \\
&\frac{|N_1|}{d_{|W_2 \cap C_1|}} = 1 = \frac{d_1}{d_1} = \frac{|N_2|}{|W_2 \cap C_2| + 1} \textrm{,} \\
&\frac{|N_1|}{d_{|W_2 \cap C_1|}} = 1 \geq \frac{d_1}{d_2} = \frac{|N_i|}{d_{|W_2 \cap C_i| + 1}} \textrm{,} \\
&\frac{|N_i|}{d_{|W_2 \cap C_i|}} = 1 = \frac{d_1}{d_1} = \frac{|N_2|}{d_{|W_2 \cap C_2| + 1}} \textrm{,} \\
&\frac{|N_i|}{d_{|W_2 \cap C_i|}} = 1 \geq \frac{d_x}{d_{x + 1}} = \frac{|N_1|}{d_{|W_2 \cap C_1| + 1}} \textrm{.}
\end{align*}
\fi

Thus, $W_1$ and $W_2$ are winning committees and hence have the same scores. Their respective scores are
\begin{align*}
\score{f}(W_1,A) & = d_x\cdot f(x-1, y) + d_1\cdot f(1, m - y - k +x) + (k-x)\cdot d_1\cdot f(1, 1)  \textrm{,} \\
\score{f}(W_2,A) & = d_x\cdot f(x, y) + d_1\cdot f(1, m - y - k +x)  + (k-x-1)\cdot d_1\cdot f(1, 1) + d_1\cdot f(0, 1) \textrm{.}
\end{align*}
Since $\score{f}(W_1,A) = \score{f}(W_2,A)$ we have
\begin{align*}
f(x, y) = f(x-1, y) + \frac{d_1}{d_x} \Big(f(1, 1) - f(0, 1)\Big) \textrm{.}
\end{align*}

\item[$d_x=\infty$:]
Now, let us move to the case when $d_x = \infty$. Let us fix a committee $W$ and consider a set $X\subseteq C$ with $|X|=y$ and $|X\cap W|=x-1$. 
We construct a party-list profile $A$ as follows: $A$ contains $\zeta$ votes that approve $X$ (intuitively, $\zeta$ is a large natural number); further for each candidate $c \in W\setminus X$, profile $A$ contains a single voter who approves $\{c\}$.
This construction requires $y+(k-x+1)$ candidates, thus it is possible since we fixed $x$ so that $k-x < m-y$.

Clearly, committee $W$ is $\mathbf{d}$-proportional.
Let $W'$ be the committee we obtain from $W$ by replacing one candidate in $W\setminus X$ with a candidate in $X\setminus W$ (such a candidate exists since $(x, y)\in D_{m,k}$).
We have $\score{f}(W',A) \leq \score{f}(W,A)$ and thus
\begin{align}
\zeta f(x, y) + (k-x) f(1, 1) + f(0, 1) < \zeta f(x-1, y) + (k-x+1) f(1, 1) \text{.}
\end{align}
The above condition can be written as $f(x, y) - f(x-1, y) \leq \frac{1}{\zeta} \cdot (f(1, 1) - f(0, 1))$. Since this must hold for any $\zeta$, we get that $f(x+1, y) \leq f(x, y)$. Since $f$ is a counting function, $f(x, y) \geq f(x-1, y)$; thus we get that $f(x, y) = f(x-1, y)$, i.e.,:
\begin{align*}
f(x, y) = f(x-1, y) + \frac{d_1}{\infty} \Big(f(1, 1) - f(0, 1)\Big) \textrm{.}
\end{align*}
(Above, we use the convention that $\frac{\infty}{\infty} = 0$.)
\end{enumerate}

Now, as we have shown that \begin{align*}
f(x, y) = f(x-1, y) + \frac{d_1}{d_x} \Big(f(1, 1) - f(0, 1)\Big) \textrm{.}
\end{align*} holds for $1 \leq x \leq k$ such that $k-x < m-y$, we can expand this equation until we reach $x = 0$ or $x = k+y-m$. Let $s(y) = \max(0, k+y-m)$.
\begin{align*}
f(x, y) &= f(s(y), y) + d_1\Big(f(1, 1) - f(0, 1)\Big)\sum_{i=s(y)+1}^x w_i \\
        &= f(s(y), y) - d_1\Big(f(1, 1) - f(0, 1)\Big)\sum_{i=1}^{s(y)}w_i  + d_1\Big(f(1, 1) - f(0, 1)\Big)\sum_{i=1}^x w_i\textrm{.}
\end{align*}
Obviously, the above equality also holds for $x = s(y)$.

Hence we have shown that indeed $f(x)=c\cdot f_{w\text{-T}}(x,y) + d(y)$ for $c=d_1(f(1, 1) - f(0, 1))$ and $d(y)=f(s(y), y) - d_1\Big(f(1, 1) - f(0, 1)\Big)\sum_{i=1}^{s(y)}w_i$.
By Proposition~\ref{prop:counting-functions-equivalent}, $\calF$ is $w$-Thiele.
\end{proof}

\begin{replemma}{lem:sym+con+prop->pareto}
\lemsymconproppareto
\end{replemma}
\begin{proof}
Let $\calF$ be an ABC ranking rule satisfying symmetry, consistency, and $\mathbf{d}$-proportionality. To show that $\calF$ satisfies weak efficiency, it suffices to show that $\calF$ satisfies weak efficiency for single-voter profiles. Indeed, 
assume that $\calF$ satisfies weak efficiency for single-voter profiles. Let $W_1, W_2 \in \pow_k(C)$ and $A\in\calA(C,V)$ where no voter approves a candidate in $W_2\setminus W_1$; we want to show that $W_1 \succeq_{\calF(A)} W_2$.
Since weak efficiency holds for single-voter profiles, we know that $W_1 \succeq_{\calF(A(v))} W_2$ for all $v\in V$.
By consistency we can infer that $W_1 \succeq_{\calF(A)} W_2$.

For the sake of contradiction let us assume that $\calF$ does not satisfy weak efficiency for single-voter profiles. This means that there exist $X\subseteq C$ and $W_1, W_2\in \pow_k(C)$ such that $(W_2\setminus W_1) \cap X = \emptyset$ and $W_2 \succ_{\calF(X)} W_1$. First, we show that in such case there exist $W \in \pow_{k-1}(C)$, $c, c' \in C$ with $c\in X$, $c'\notin X$, and $W \cup \{c'\} \succ_{\calF(X)} W \cup \{c\}$. Let $z = |W_1 \cap X| - |W_2 \cap X|$, and let us consider the following sequence of $z$ operations which define $z$ new committees. We start with committee $W_{2, 1} = W_2$, and in the $i$-th operation, $i \in [z-1]$, we construct $W_{2, i+1}$ from $W_{2, i}$ by removing from $W_{2, i}$ one arbitrary candidate in $W_{2, i}\setminus X$ and by adding one candidate from $(W_1\setminus W_2) \cap X$. Consequently, $|W_{2, z} \cap X| = |W_1\cap X|$, so by neutrality we have $W_{2, z} \sim_{\calF(X)} W_1$. By our assumption we have that $W_{2, 1} \succ_{\calF(X)} W_{2, z}$, thus, there exists $i \in [z-1]$ such that $W_{2, i} \succ_{\calF(X)} W_{2, i+1}$. The committees $W_{2, i}$ and $W_{2, i+1}$ differ by one element only, so we set $W = W_{2, i} \cap W_{2, i+1}$, $c \in W_{2, i+1} \setminus W_{2, i}$ and $c' \in W_{2, i} \setminus W_{2, i+1}$, and we have $W \cup \{c'\} \succ_{\calF(X)} W \cup \{c\}$ for $c\in X$ and $c'\notin X$.

Let $\ell$ denote the number of members of $W \cup \{c\}$ which are approved in $X$, i.e., $\ell=|(W\cup \{c\})\cap X|$.
Let us consider the following party-list profile $A'$. There are two groups of voters: $N_1$ with $|N_1|=d_{\ell}$ and $N_2$ with $|N_2|= d_{k - \ell}$.
The voters in $N_1$ approve of $X$; the voters in $N_2$ approve $C \setminus (X \cup \{c'\})$. From $\mathbf{d}$-proportionality we infer that committee $W \cup \{c\}$ is winning: 
\begin{align*}
&\frac{|N_1|}{d_{|(W \cup \{c\}) \cap X|}} = 1 \geq \frac{d_{k-\ell}}{d_{k-\ell + 1}} = \frac{|N_2|}{d_{|(W\cup \{c\}) \cap (C \setminus (X \cup \{c'\})| + 1}}, \\
&\frac{|N_2|}{d_{|(W\cup \{c\}) \cap (C \setminus (X \cup \{c'\})|}} = 1 \geq \frac{d_\ell}{d_{\ell+1}} = \frac{|N_1|}{d_{|(W \cup \{c\}) \cap X|+1}}. 
\end{align*}
This, however, yields a contradiction: Voters from $N_1$ prefer $W \cup \{c'\}$ over $W \cup \{c\}$ since $W \cup \{c'\} \succ_{\calF(X)} W \cup \{c\}$. For voters from $N_2$ committees $W \cup \{c'\}$ and $W \cup \{c\}$ are equally good by neutrality. Hence, by consistency, it holds that $W \cup \{c'\} \succ_{\calF(A')} W \cup \{c\}$, a contradiction.
We conclude that $W_1 \succeq_{\calF(X)} W_2$ and hence weak efficiency holds for single-voter profiles and---in consequence---for arbitrary profiles.
\end{proof}

\begin{repproposition}{prop:lowerquota}\proplowerquota
\end{repproposition}
\begin{proof}
Consider a party-list profile $A$ with one group of voters $N_1$ approving $y$ candidates and $k-x+1$ groups of voters, $N_2, \ldots, N_{k-x+2}$, each approving a single candidate---for each $i \in [k-x+2]$ let $C_i$ denote the set of candidates approved by voters from $N_i$. Each of the remaining $m - y - k + x - 1$ candidates is not approved by any voter. We set $|N_1| = x(k-x+1)$, and for each $i \geq 2$ we set $|N_i| = k-x$. Observe that:
\begin{align*}
k \cdot \frac{|N_1|}{|V|} = k \cdot \frac{x(k-x+1)}{x(k-x+1) + (k-x+1)(k-x)} = k \cdot \frac{x(k-x+1)}{k(k-x+1)} = x \text{.}  
\end{align*} 
From the lower quota property, we infer that there exists a winning committee $W$ such that
$|W \cap C_1| \geq x$, and from the pigeonhole principle we get that there exists $i \geq 2$ with $W \cap C_i = \emptyset$; let $C_i = \{c_i\}$. Thus, the score of committee $W$ is higher than or equal to the score of committee $(W \cup \{c_i\}) \setminus \{c\}$ for $c \in W \cap C_1$. As a result we get that $f(x, y) |N_1| \geq f(x-1, y) |N_1| + f(1, 1)|N_i|$, which can be equivalently written as:
\begin{align*}
f(x, y) \geq f(x-1, y) + \frac{1}{x} \cdot f(1, 1) \cdot \frac{k-x}{k-x+1} \text{.} 
\end{align*}
Now, consider another similar party-list profile, with the only difference that $|N_1| = x-1$, and $|N_i| = 1$ for $i \geq 2$. Observe that for $i \geq 2$:
\begin{align*}
k \cdot \frac{|N_i|}{|V|} = k \cdot \frac{1}{x-1 + (k-x+1)} = 1 \text{.}  
\end{align*}
Thus, for each $i \geq 2$ we have that $|W \cap C_i| = 1$. By a similar reasoning as before we get that: $f(1, 1)|N_i| + f(x-1, y)|N_1| \geq f(x, y) |N_1|$, which is equivalent to:
\begin{align*}
f(x, y) \leq f(x-1, y) + \frac{1}{x-1} \cdot f(1, 1) \text{.} 
\end{align*}
This completes the proof.
\end{proof}

\section{ABC Choice Rules}
\label{sec:app:choice_rules}

As noted earlier, every ABC ranking rule $\calF$ induces an ABC choice rule by selecting the top-ranked committees in the weak order returned by $\calF$.
From a mathematical point of view, ABC choice rules are quite different from ABC ranking rules since, in particular, losing committees under ABC choice rules are not distinguishable. Thus, obtaining an axiomatic characterization of an ABC choice rule might require a different approach than the one used for finding a characterization of a related ABC ranking rule.
This is also reflected in the literature on axiomatic characterization of single-winner voting rules, where social welfare functions and social choice functions have been usually considered separately, and corresponding characterizations often required considerably different proofs (cf.\ the characterization of positional scoring rules~\citep{young74,you:j:scoring-functions}). 

In this section we present a technique that allows us to directly translate some of our previous results for ABC ranking rules to ABC choice rules. 
In particular we show that the axiomatic characterizations of PAV and CC can be transferred to the setting of ABC choice rules. %
We start by formulating the relevant axioms from Section~\ref{subsec:basic-axioms} so as to be applicable to ABC choice rules.

\axiomset{
  \textbf{Anonymity.} An ABC choice rule $\calR$ is \emph{anonymous} if for
  every two (not necessarily different) sets of voters $V, V' \subseteq
  \naturals$ such that $|V| = |V'|$, for each bijection $\rho: V \to V'$ and
  for every two approval preference profiles $A \in \calA(C,V)$ and $A' \in \calA(C,V')$
  such that $A(v) = A'(\rho(v))$ for each $v \in V$,
  it holds that $\calR(A) = \calR(A')$.

\medskip\noindent
  \textbf{Neutrality.} An ABC choice rule $\calR$ is \emph{neutral} if for each permutation
  $\sigma$ of $C$ and every two approval preference profiles $A, A' \in \calA(C,V)$
  over the same voter set $V$ with $\sigma(A) = A'$ it holds that $\{\sigma(W)\colon W\in \calR(A)\} = \calR(A')$.

\medskip\noindent
  \textbf{Consistency.} An ABC choice rule $\calR$ is \emph{consistent} if for every two
  profiles $A$ and $A'$ over disjoint sets of voters, $V \subset \naturals$
  and $V' \subset \naturals$, $V \cap V' = \emptyset$, if $\calR(A) \cap \calR(A') \neq \emptyset$ then $\calR(A + A') = \calR(A) \cap \calR(A')$.

\medskip\noindent
  \textbf{Weak efficiency.} An ABC choice rule $\calR$ satisfies \emph{weak efficiency} if for each approval profile $A$, each winning
  committee $W \in \calR(A)$, each candidate $c \in W$ who is not approved by any voter, and each candidate $c' \notin W$ it holds that $(W\setminus \{c\}) \cup \{c'\} \in \calR(A)$. 

\medskip\noindent
  \textbf{Continuity.} An ABC choice rule $\calR$ satisfies \emph{continuity} if for
  every two approval profiles $A$ and $A'$ with $\calR(A) \cap \calR(A') = \emptyset$, there exists a number $n \in
  \naturals$ such that %
  $\calR(A+nA') \subseteq \calR(A')$.

}

Furthermore, note that D'Hondt proportionality and disjoint diversity only apply to \emph{winning} committees and hence can be used for ABC choice rules without modification.
Let us first recall the axiom of 2-Nonimposition.

\axiom{2-Nonimposition}{An ABC choice rule $\calR$ satisfies \emph{2-Nonimposition} if for
 every two committees $W_1, W_2 \in \pow_k(C)$ there exists an approval profile $\alpha(W_1, W_2)$ such that $\calR(\alpha(W_1, W_2)) = \{W_1, W_2\}$.
}

\subsection{Proof of Theorem~\ref{thm:characterizationChoiceFunctions}}

Let us fix $\calR$ to be an ABC choice rule which satisfies symmetry, consistency, continuity, and 2-Nonimposition.
We will show that $\calR$ uniquely defines a corresponding ABC ranking rule $\calF_{\calR}$ and that $\calF_{\calR}$ preservers many axiomatic properties. This observation will allow us to apply our previous results to ABC choice rules. Let $\alpha$ be a fixed function from $\pow_k(C)\times \pow_k(C)$ to $\calA(C)$ such that for every two committees $W_1, W_2 \in \pow_k(C)$ it holds that $\calR(\alpha(W_1, W_2)) = \{W_1, W_2\}$. Such a function exists because $\calR$ satisfies 2-Nonimposition.
We define $\calF_{\calR}$ as follows:
\begin{definition}
For each $A\in\calA(C,V)$ we define $\calF_{\calR}(A)$ so that for each $W_1,W_2\in \pow_k(C)$,
\begin{align*}
W_1 \succeq_{\calF_{\calR}(A)} W_2 \iff \exists_{n}\forall_{n' \geq n} \;\; W_1 \in \calR(A + n'\alpha(W_1, W_2)) \textrm{.}
\end{align*}
\label{def:calF_calR}
\end{definition}
As a consequence of Definition~\ref{def:calF_calR} we have
\begin{align*}
W_1 \succ_{\calF_{\calR}(A)} W_2 \iff \exists_{n}\forall_{n' \geq n} \;\; \calR(A + n'\alpha(W_1, W_2))=\{W_1\} \textrm{.}
\end{align*}

The definition of $\calF_{\calR}$ seemingly depends on the choice of $\alpha$. We show that this is not the case.
\begin{lemma}
Let $\alpha,\alpha'$ be functions from $\pow_k(C)\times \pow_k(C)$ to $\calA(C)$ such that $\calR(\alpha(W_1, W_2))=\calR(\alpha'(W_1, W_2))=\{W_1,W_2\}$ for any $W_1,W_2\in \pow_k(C)$.
For every $W_1,W_2\in \pow_k(C)$ and $A\in \calA(C,V)$,
\begin{align*}
\exists_{s}\forall_{s' \geq s} \;\; W_1 \in \calR(A + s'\alpha(W_1, W_2)) \iff \exists_{t}\forall_{t' \geq t} \;\; W_1 \in \calR(A + t'\alpha'(W_1, W_2)) \textrm{.}%
\end{align*}\label{lem:calF_calR-well-defined}%
\end{lemma}
\begin{proof}
If $W_1\in \calR(A)$, then by consistency $W_1\in \calR(A + s'\alpha(W_1, W_2))$ and $W_1 \in \calR(A + t'\alpha'(W_1, W_2))$.
If $W_1\notin \calR(A)$, then we can apply continuity and see that the equivalence only fails if 
$\calR(A + s'\alpha(W_1, W_2)) = \{W_1\}$ and $\calR(A + t'\alpha'(W_1, W_2))=\{W_2\}$ (or vice versa).
By consistency, 
\begin{align*}
&\calR((A + s'\alpha(W_1, W_2)) + t'\alpha'(W_1, W_2))=\{W_1\} \text{ and }\\
&\calR((A + t'\alpha'(W_1, W_2)) + s'\alpha(W_1, W_2))=\{W_2\}.
\end{align*}
This contradicts anonymity.
\end{proof}

The following lemma now shows that the relation defined by Definition~\ref{def:calF_calR} is complete.
\begin{lemma}\label{lem:calF_calR_complete}
For every $W_1,W_2\in \pow_k(C)$ at least on of the following two conditions holds: 
\begin{enumerate}
\item
$\exists_{n}\forall_{n' \geq n} \;\; W_1 \in \calR(A + n'\alpha(W_1, W_2))$,
\item
$\exists_{n}\forall_{n' \geq n} \;\; W_2 \in \calR(A + n'\alpha(W_1, W_2))$.
\end{enumerate}
\end{lemma}
\begin{proof}
If $W_1 \in \calR(A)$, then by consistency the first condition is satisfied; if $W_2 \in \calR(A)$, then the second condition is satisfied. Let us consider what happens if $W_1, W_2 \notin \calR(A)$. By continuity, we know that there must exist an $n \in \naturals$ such that $\calR(A + n\alpha(W_1, W_2)) \subseteq \calR(\alpha(W_1, W_2)) = \{W_1, W_2\}$. Thus, $W_1 \in \calR(A + n\alpha(W_1, W_2))$, or $W_2 \in \calR(A + n\alpha(W_1, W_2))$ holds. Without loss of generality, let us assume that $W_1 \in \calR(A + n\alpha(W_1, W_2))$. Then, by consistency, for each $n' \geq n$ it holds that:
$W_1 \in \calR((A + n\alpha(W_1, W_2) + (n'-n)\alpha(W_1, W_2)) = \calR((A + n'\alpha(W_1, W_2)) \textrm{.}$
\end{proof}

\begin{lemma}\label{thm:propertiesPreserved}
If $\calR$ satisfies symmetry, then so does
$\calF_{\calR}$ (cf.~Definition~\ref{def:calF_calR}). The same holds for consistency, weak efficiency, and continuity.
\end{lemma}
\begin{proof}
(Anonymity)   Let $V, V' \subset \naturals$ such that $|V| = |V'|$.
  Further, let $A \in \calA(C,V)$ and $A' \in \calA(C,V')$ so that $A'$ can be obtained from $A$ by permuting its votes.
  We have to show that for all $W_1,W_2\in \pow_k(C)$, $W_1\succeq_{\calF_{\calR}(A)} W_2\iff W_1\succeq_{\calF_{\calR}(A')} W_2$.
  This follows from the fact that $\calR(A + n'\alpha(W_1, W_2))=\calR(A' + n'\alpha(W_1, W_2))$ by anonymity of $\calR$.

(Neutrality)  Let $\sigma$ be a permutation
  of $C$ and let $A, A' \in \calA(C,V)$ such that $\sigma(A) = A'$.
  We have to show that for all $W_1,W_2\in \pow_k(C)$, $W_1\succeq_{\calF_{\calR}(A)} W_2\iff \sigma(W_1)\succeq_{\calF_{\calR}(A')} \sigma(W_2)$, i.e., $\exists_{n}\forall_{n' \geq n} \;\; W_1 \in \calR(A + n'\alpha(W_1, W_2))\iff \exists_{n}\forall_{n' \geq n} \;\; \sigma(W_1) \in \calR(A' + n'\alpha(\sigma(W_1), \sigma(W_2)))$. By Lemma~\ref{lem:calF_calR-well-defined} and neutrality of $\calR$, $\exists_{n}\forall_{n' \geq n} \;\; \sigma(W_1) \in \calR(A' + n'\alpha(\sigma(W_1), \sigma(W_2)))\iff \exists_{n}\forall_{n' \geq n} \;\; \sigma(W_1) \in \calR(\sigma(A) + \sigma(n'\alpha(W_1,W_2))$.
  Again by neutrality of $\calR$, we have $\exists_{n}\forall_{n' \geq n} \;\; \sigma(W_1) \in \calR(\sigma(A) + \sigma(n'\alpha(W_1,W_2)) \iff \exists_{n}\forall_{n' \geq n} \;\; \sigma(W_1) \in \sigma(\calR(A + n'\alpha(W_1,W_2)) \iff \exists_{n}\forall_{n' \geq n} \;\; W_1 \in \calR(A + n'\alpha(W_1, W_2))$.

(Consistency) Let us first prove Statement (i) from the definition of consistency and for this let $W_1, W_2\in \pow_k(C)$ with $W_1\succ_{\calF_{\calR}(A)} W_2 $ and $W_1\succeq_{\calF_{\calR}(A')} W_2 $.
Due to the fact that $W_1\succ_{\calF_{\calR}(A)} W_2 $ and $W_1\succeq_{\calF_{\calR}(A')} W_2 $ and by Definition~\ref{def:calF_calR}, there exists an $n$ with $\calR(A + n'\alpha(W_1, W_2))=\{W_1\}$ for all $n'\geq n$ and $W_1\in \calR(A' + n'\alpha(W_1, W_2))$ for all $n'\geq n$.
Since $\calR(A + n'\alpha(W_1, W_2))\cap\calR(A' + n'\alpha(W_1, W_2))\neq \emptyset$, consistency of $\calR$ implies that $\calR(A + n'\alpha(W_1, W_2) + A' + n'\alpha(W_1, W_2))=\{W_1\}$, for all $n' \geq n$.
By anonymity, $\calR(A + A' + 2n'\alpha(W_1, W_2))=\{W_1\}$ for all $n'\geq n$.
Hence $W_1\succ_{\calF_{\calR}(A+A')} W_2$.
Statement (ii) can be shown analogously except that it suffices to show that $W_1\in \calR(A + A' + 2n'\alpha(W_1, W_2))$.

(Weak efficiency) Let $W_1, W_2 \in \pow_k(C)$ and $A\in\calA(C,V)$ such that no voter approves a candidate in $W_2\setminus W_1$. If $\exists_{n}\forall_{n' \geq n} \;\; W_1 \in \calR(A + n'\alpha(W_1, W_2))$, then $W_1 \succeq_{\calF_{\calR}(A)} W_2$. Otherwise, by Lemma~\ref{lem:calF_calR_complete} we know that $\exists_{n}\forall_{n' \geq n} \;\; W_2 \in \calR(A + n'\alpha(W_1, W_2))$. Since $\calR$ satisfies weak efficiency, we get that $W_2 \in \calR(A + n'\alpha(W_1, W_2))$ implies $W_1 \in \calR(A + n'\alpha(W_1, W_2))$, and so by the definition of $\calF_{\calR}$ we get that $W_1 \succeq_{\calF_{\calR}(A)} W_2$.

(Continuity) Let $W_1, W_2\in \pow_k(C)$ and let $A$ and $A'$ be two approval profiles with $\calR(A) \cap \calR(A') = \emptyset$, and with $W_1 \succ_{\calF_{\calR}(A')} W_2$. Thus, for sufficiently large $n$, for each $n' \geq n$, it holds that $\calR(A' + n'\alpha(W_1, W_2)) = \{W_1\}$. Since $\calR$ satisfies continuity, for sufficiently large $n''$, for each $n''' \geq n''$ it holds that $\calR(A + n'''(A' + n'\alpha(W_1, W_2))) = \{W_1\}$. This proves that $W_1 \succ_{\calF_{\calR}(A + n'''A')} W_2$ and thus continuity of $\calF_{\calR}$.
\end{proof}

We now show that winners selected by $\calR$ and by $\calF_\calR$ are the same.

\begin{lemma}\label{lem:winners-same}
Let $A\in\calA(C,V)$ and $W\in\pow_k(C)$. It holds that $W\in \calR(A)$ if and only if $W$ is a winning committee in $\calF_\calR(A)$.
\end{lemma}

\begin{proof}
Let $W\in \calR(A)$; we will show that $W$ is maximal in $\calF_\calR(A)$. Let $W'\in \pow_k(C)$. By consistency of $\calR$, $W\in \calR(A + n'\alpha(W,W'))$ for any $n'$. Hence $W\succeq_{\calF_\calR(A)} W'$. Since this holds for every committee $W'$, we infer that $W$ is a winning committee.

Let $W$ be a winning committee in $\calF_\calR(A)$; we will show that $W \in \calR(A)$. Let $W'\in \pow_k(C)$ and towards a contradiction assume that $W' \in \calR(A)$ but $W \notin \calR(A)$. 
Since $W$ is a winning committee in $\calF_\calR(A)$, it holds that $\exists_{n}\forall_{n' \geq n} \; W\in \calR(A + n'\alpha(W, W'))$.
However, by consistency, $\calR(A + n'\alpha(W, W'))=\{W'\}$, a contradiction.
\end{proof}

\begin{reptheorem}{thm:characterizationChoiceFunctions}
\thmcharacterizationChoiceFunctions
\end{reptheorem}
\begin{proof}
It is easy to verify that the ABC counting rules satisfy the axioms of symmetry, consistency, weak efficiency, and continuity as formulated for ABC choice rules. To prove the other direction, let $\calR$ be a function that satisfies symmetry, consistency, weak efficiency, continuity and 2-Nonimposition. 
Consequently, using Definition~\ref{def:calF_calR} we can construct the ABC ranking rule $\calF_{\calR}$. By Lemma~\ref{thm:propertiesPreserved}, $\calF_{\calR}$ satisfies symmetry, consistency, weak efficiency, and continuity. 
By Lemma~\ref{lem:winners-same}, the winning committees in $\calF_\calR$ are the same as winning committees in $\calR$.
By Theorem~\ref{thm:characterizationWelfareFunctions} we infer that $\calF_{\calR}$ is an ABC counting rule.
By Lemma~\ref{lem:winners-same} we get that $\calR$ has exactly the same winning committees as $\calF_{\calR}$, and so we infer that $\calR$ is an ABC counting rule. 
\end{proof}

\subsection{Proving Characterizations of Specific ABC Choice Rules}

\begin{replemma}{lem:proportionality_implies_2_nonimposition}
\lemproportionalityimpliestwononimposition
\end{replemma}
\begin{proof}
Let us fix two committees, $W_1$ and $W_2$, with $W_1 \neq W_2$. For every two candidates, $c_1$ and $c_2$, with $c_1 \in W_1 \setminus W_2$ and  $c_2 \in W_2 \setminus W_1$ we construct the profile $\beta(c_1, c_2)$ in the following way. In $\beta(c_1, c_2)$ there is one voter who approves $c_1$ and $c_2$. Further for each candidate $c \in W_1 \cup W_2$ with $c \notin \{c_1, c_2\}$ we introduce one voter who approves of $c$. Note that $\beta(c_1, c_2)$ is a party-list profile. According to $\mathbf{d}$-proportionality, each size-$k$ committee that consists of candidates from $W_1 \cup W_2$ and that does not contain both $c_1$ and $c_2$ is winning in $\beta(c_1, c_2)$.
This is due to our assumption that $d_2>d_1$.
Observe that both $W_1$ and $W_2$ are winning committees in $\beta(c_1, c_2)$.
 Now, let us consider the profile:
\begin{align*}
\alpha(W_1, W_2) = \sum_{c_1 \in W_1 \setminus W_2, c_2 \in W_2 \setminus W_1} \beta(c_1, c_2) \text{.}
\end{align*}
By consistency, $W_1$ and $W_2$ are the only winning committees in $\alpha(W_1, W_2)$, which completes the proof. 
\end{proof}

\begin{reptheorem}{thm:pav-abc-winner-rule-characterization}
\thmpavabcwinnerrulecharacterization
\end{reptheorem}

\begin{proof}
The theorem follows from Lemma~\ref{lem:proportionality_implies_2_nonimposition} and Theorem~\ref{thm:characterizationChoiceFunctions}.
\end{proof}

\begin{lemma}\label{lem:diversity_implies_2_nonimposition}
An ABC choice rule that satisfies symmetry, consistency and disjoint diversity also satisfies 2-Nonimposition.
\end{lemma}
\begin{proof}
We fix two committees $W_1$ and $W_2$, $W_1 \neq W_2$, and construct a profile $\alpha(W_1, W_2)$ in the following way. 
Let $\calM(W_1, W_2)$ denote the set of bijections from $W_1 \setminus W_2$ to $W_2 \setminus W_1$.
Fix $m \in \calM(W_1, W_2)$ and let us construct profile $\beta_m(W_1, W_2)$ in the following way: For each $c\in W_1\setminus W_2$, we introduce one voter who approves of $\{c, m(c)\}$; further, for each candidate from $c \in W_1 \cap W_2$ we introduce one voter who approves $\{c\}$. From disjoint diversity and from symmetry we get that each committee that contains all candidates from $W_1 \cap W_2$ and that for each matched pair $(c, m(c))$ contains either $c$ or $m(c)$ (but not both of them), is winning.   

Now, we construct the profile $\alpha(W_1, W_2)$ as follows:
\begin{align*}
\alpha(W_1, W_2) = \sum_{m \in \calM(W_1, W_2)} \beta_m(W_1, W_2) \text{.}
\end{align*}
It follows from consistency that $W_1$ and $W_2$ are the only two winning committees.
\end{proof}

\begin{reptheorem}{thm:cc-abc-winner-rule-characterization}
\thmccabcwinnerrulecharacterization
\end{reptheorem}
\begin{proof}
The same reasoning as in the proof of Theorem~\ref{thm:characterizationChoiceFunctions} can be applied.
\end{proof}

\begin{repproposition}{prop:MAVfails2nonimposition}
\propMAVfailstwononimposition
\end{repproposition}
\begin{proof}
Let us take two committees, $W_1$ and $W_2$, such that $|W_1 \setminus W_2| = |W_2 \setminus W_1| \geq 2$. Consider a profile $A$ where $W_1$ and $W_2$ are unique winners. This means that each candidate from $W_1 \setminus W_2$ has the same approval score as each candidate from $W_2 \setminus W_1$. Indeed, if the approval score of some candidate $c \in W_1 \setminus W_2$ were higher then the approval score of some candidate $c' \in W_2 \setminus W_1$, then $(W_2 \setminus \{c'\}) \cup \{c\}$ would be a better committee than $W_2$, and so $W_2$ would not be winning. But this means that for each $c \in W_1 \setminus W_2$ and each $c' \in W_2 \setminus W_1$, the committee $(W_1 \setminus \{c\}) \cup \{c'\}$ is as good as $W_1$ according to AV, thus it is also a winner. Consequently, $W_1$ and $W_2$ are not unique winners.
\end{proof}

\bibliography{../main}

\end{document}